\documentclass[english,a4paper,11pt]{article}

\usepackage{amssymb}
\usepackage{amsthm}

\usepackage[utf8]{inputenc}
\usepackage[T1]{fontenc}
\usepackage[margin=30mm]{geometry}

\usepackage[figuresright]{rotating}
\usepackage{authblk}

\theoremstyle{thmstyleone}%
\newtheorem{theorem}{Theorem}

\newtheorem{lemma}{Lemma}
\newtheorem{proposition}[theorem]{Proposition}% 

\theoremstyle{thmstyletwo}%
\newtheorem{remark}{Remark}

\theoremstyle{thmstylethree}%

\newenvironment{hyp}[1]{
  \begin{enumerate}[label=\textbf{\sf(#1\arabic*)},resume=hyp#1]\begin{sf}}
{\end{sf}\end{enumerate}}

\newenvironment{hypn}[1]{
  \begin{enumerate}[label=\textbf{\sf(#1)}]\begin{sf}}
{\end{sf}\end{enumerate}}

\newcounter{notecounter}

\newcommand{\pik}[1]{\pi_0^{#1}}

\newcommand{\abs}[1]{\left\lvert#1\right\rvert}
\newcommand{\absm}[1]{\lvert#1\rvert}

\newcommand{\convasbis}{\ \overset{\PP^{\bar Q}_{\xip}-a.s.}{{\longrightarrow}}\ }
\newcommand{\convaszeta}{\ \overset{\PP^{\bar Q}_{\zeta}-a.s.}{{\longrightarrow}}\ }
\newcommand{\convergesto}[2]{\underset{#1\rightarrow #2}{\longrightarrow}}
\newcommand{\convproba}[1]{\ \overset{#1-\mathrm{prob}}{{\longrightarrow}}\ }
\newcommand{\convlaw}[1]{\overset{#1-law}{\leadsto}}

\newcommand{\dd}[1]{\mathrm{d}#1}
\renewcommand{\det}{\mathrm{det}}

\newcommand{\EE}{\mathbb E}
\newcommand{\ens}[1]{\mathsf{#1}}

\newcommand{\eqdef}{:=}

\newcommand{\eqsp}{\,}

\newcommand{\ess}{\mathrm{ESS}}

\newcommand{\eset}{\mathbb{E}}
\newcommand{\espcond}[2]{\mathbb{E} \left[ #1 \middle \vert #2 \right]}
\newcommand{\espcondmkv}[4]{\mathbb{E}_{#1}^{#2} \left[ #3 \middle \vert #4 \right]}

\newcommand{\Fcal}{\mathcal{F}}
\newcommand{\floor}[1]{\left\lfloor #1\right\rfloor }
\newcommand{\fracpart}[1]{\left\langle #1 \right\rangle }

\newcommand{\funcset}[2]{\mathsf{F}_{#2}(#1)}
\renewcommand{\geq}{\geqslant}

\newcommand{\indi}[1]{\mathbf{1}_{#1}}
\newcommand{\indiacc}[1]{\mathbf{1}_{\{#1\}}}

\renewcommand{\leq}{\leqslant}

\newcommand{\lr}[1]{\left(#1 \right)}

\newcommand{\lrb}[1]{\left[#1 \right]}
\newcommand{\lrbm}[1]{[#1]}
\newcommand{\lrc}[1]{\left\{#1 \right\}}
\newcommand{\lrcb}[1]{\left\{#1 \right\}}

\newcommand{\mcbb}{\mathcal B}
\newcommand{\mcf}{\mathcal{F}}

\newcommand{\measureset}{\mathsf{M}}
\newcommand{\Ncal}{\mathcal{N}}
\newcommand{\norm}[1]{\left \| #1 \right \|}
\newcommand{\nset}{\mathbb N}
\newcommand{\nsetzero}{{\nset_0}}
\newcommand{\NN}{\mathbb{N}}

\newcommand{\piaux}{\mu}
\newcommand{\posfunc}[1]{\mathsf{F}_+(#1)}

\newcommand{\PE}{\mathbb E}
\newcommand{\PP}{\mathbb P}
\newcommand{\Prob}{\mathbb{P}}

\newcommand{\rej}{S}
\newcommand{\rmd}{\mathrm d}
\newcommand{\rme}{\mathrm e}

\newcommand{\Ropt}{R_{\scriptscriptstyle \mathrm{opt}}}
\newcommand{\Ropttilde}{\tilde R_{\scriptscriptstyle \mathrm{opt}}}
\newcommand{\Rpm}{R_{\scriptscriptstyle \mathrm{pm}}}
\newcommand{\Rpmtilde}{\tilde R_{\scriptscriptstyle \mathrm{pm}}}
\newcommand{\rset}{\mathbb{R}}

\newcommand{\seqq}[2]{\{#1: \eqsp #2\}}
\newcommand{\sett}[2]{\lrc{#1: \eqsp #2}}

\newcommand{\weaklim}[1]{\overset{#1}{\rightsquigarrow}}
\newcommand{\xip}{{\xi'}}
\newcommand{\Xset}{{\mathsf X}}
\newcommand{\Xaux}{\tilde{X}}

\newcommand{\Xsigma}{\mathcal X}
\newcommand{\Yset}{{\mathsf Y}}
\newcommand{\Ysigma}{\mathcal Y}
\newcommand{\Var}{\mathbb{V}\mathrm{ar}}

\usepackage{amsmath}
\usepackage{algorithm}
\usepackage{algpseudocode}
\usepackage{xcolor}
\usepackage{framed}
\usepackage[colorlinks]{hyperref}
\usepackage{cleveref}
\usepackage{enumitem}
\usepackage{ulem}
\usepackage{comment}
\usepackage{booktabs}
\usepackage{multirow}
\begin{document}



\title{The Importance Markov Chain}


\author[1]{Charly Andral}%\ead{andral@ceremade.dauphine.fr}

\author[2]{Randal Douc}%\ead{randal.douc@telecom-sudparis.eu}

\author[2]{Hugo Marival}%\ead{hugo.marival@telecom-sudparis.eu}

\author[1,3]{Christian P. Robert}%\ead{xian@ceremade.dauphine.fr}

\affil[1]{
    {CEREMADE, CNRS, UMR 7534, Université Paris-Dauphine, PSL University}, 75016, Paris, France
}

\affil[2]{
    {SAMOVAR, Telecom Sudparis, Institut
    Polytechnique de Paris}, 9 rue Charles
    Fourier, {91820, Evry}, France
}

\affil[3]{
    {Department of Statistics, University of Warwick}, Coventry, {CV4 7AL}, UK
}

\affil[ ]{Emails: \{andral, xian\}@ceremade.dauphine.fr, \{randal.douc, hugo.marival\}@telecom-sudparis.eu}

\maketitle

\begin{abstract}
    The Importance Markov chain is a novel algorithm bridging the gap between rejection sampling and importance sampling, moving from one to the other through a tuning parameter. Based on a modified sample of an instrumental Markov chain targeting an instrumental distribution (typically via a MCMC kernel), the Importance Markov chain produces an extended Markov chain where the marginal distribution of the first component converges to the target distribution. For example, when targeting a multimodal distribution, the instrumental distribution can be chosen as a tempered version of the target which allows the algorithm to explore its modes more efficiently.
    We obtain a Law of Large Numbers and a Central Limit Theorem as well as geometric ergodicity for this extended kernel under mild assumptions on the instrumental kernel. Computationally, the algorithm is easy to implement and preexisting libraries can be used to sample from the instrumental distribution.

\end{abstract}
\providecommand{\keywords}[1]
{
  \small	
  \textbf{\textit{Keywords---}} #1
}


\section{Introduction}\label{sec1}

In Monte Carlo methods \cite{metropolisMonteCarloMethod1949} and in particular in computational Bayesian statistics, sampling is used to construct estimates for quantities depending on problem-specific distributions. As a first approach, one can simulate independently according to another
distribution, called the \textit{instrumental} 
distribution, and use this sample to build an estimate of the quantity of interest (see, e.g.,  \cite{robertMonteCarloStatistical2010} for a general introduction to Monte Carlo methods). The most well-known example is the importance sampling (IS) technique \cite{kahnModificationMonteCarlo1950}, which produces a weighted
sample to approximate $\pi(f) = \int f(x)\pi(x) \dd x$ where $\pi$ is a given distribution (by an abuse of
notation, we also denote $\pi$ its density with respect to a dominating measure $\rmd x$). Importance sampling is based on rewriting the quantity of interest as $\pi(f)=\int f(x) \frac{\pi(x)}{\tilde \pi (x)} \tilde \pi(x) \dd x$ for any density $\tilde \pi$ that dominates $\pi$. Then, $\pi(f)$ can be estimated by sampling independently $X_1,X_2,\cdots$ from the instrumental distribution $\tilde \pi$ and by returning the estimate $\tilde I = n^{-1} \sum_{i=1}^n \frac{\pi(X_i)}{\tilde \pi(X_i)} f(X_i)$. It is fundamental to recall here that importance sampling does not deliver a sample distributed from $\pi$. 
In contrast, rejection sampling allows to construct a perfect sample according to $\pi$ but at the cost that a portion of the sampled points are rejected. To be more specific, if we assume that $\pi \leq M \tilde \pi$ for some constant $M$, then we sample independently $X_1,X_2, \cdots \sim_{iid} \tilde \pi$ and $U_1, U_2, \cdots \sim_{iid} \mathcal{U}(0,1)$ until the condition $U_i < \frac{\pi(X_i)}{M \tilde \pi(X_i)}$ is met. For the exit index $i$, setting $Y=X_i$, it turns out that the law of the accepted candidate $Y$ is then exactly  $\pi$ (see \cite{devroyeNonuniformRandomVariate1986}).

Another approach is proposed by Markov Chain Monte Carlo (MCMC) methods: instead of constructing an independent and identically distributed (iid) sample, an MCMC algorithm provides a Markov chain (thus a dependent sample), that converges to the distribution of interest. The most common MCMC algorithm is the Metropolis-Hastings algorithm \cite{metropolisEquationStateCalculations1953,hastingsMonteCarloSampling1970}. 
Note that MCMC and IS are not incompatible, and the idea of using a Markov chain for the \textit{instrumental} distribution appeared as soon as 1963 \cite{fosdickMonteCarloComputations1963}, and is mentioned by Hasting in \cite{hastingsMonteCarloSampling1970}. 
More recently, many algorithms combine IS and MCMC \cite{botevMarkovChainImportance2013,raicescruzIterativeImportanceSampling2022,schusterMarkovChainImportance2021a}.  

The Importance Markov Chain (IMC) algorithm uses those ideas in a novel way.
Indeed, while most MCMC algorithms try to adjust the proposals to explore the support of the target distribution efficiently, the IMC algorithm allows to target a more friendly \textit{instrumental} distribution which is then transformed into the initial target with IS. More specifically, the \textit{instrumental} Markov chain is transformed into an \textit{augmented} Markov chain targetting the distribution of interest on its first marginal. This is different from classic subsampling or thinning of the chain that preserve the distribution \cite{maceachernSubsamplingGibbsSampler1994,linkThinningChainsMCMC2012,owenStatisticallyEfficientThinning2017}. The instrumental distribution is considered as a given. Indeed, our aim in this paper is to establish properties that are preserved by our transformation for a given \textit{instrumental} distribution---namely a law of large numbers (LLN), a Central Limit Theorem (CLT) and geometric ergodicity. 

Of course, adding a resampling step to a classical importance sampling based on a $\tilde \pi$-sample $(\tilde X_1,\ldots,\tilde X_n)$ may lead to a random variable with distribution $ \hat \pi_n $ close to the target distribution $\pi$. But the total variation norm between the two distributions $ \hat \pi_n $ and $\pi$ is typically of order $O(1/n)$ whereas our Importance Markov chain, under mild assumptions, is geometrically ergodic, showing that the decrease in the total variation norm may be geometrically fast with respect to $n$. 

The Importance Markov chain, in the specific setting with independent proposals, is related to previous works on \textit{Self Regenerative Markov chains} \cite{sahuSelfregenerativeMarkovChain2003,gasemyrMarkovChainMonte2002} and on Dynamic Weighting Monte Carlo \cite{wongDynamicWeightingMonte1997,liuTheoryDynamicWeighting2001}. It was developed in the dependent case in \cite{malefakiConvergenceProperlyWeighted2008} but the framework there was restrained to a semi-Markov formulation. 

The article is organized as follows:
\begin{enumerate}
    \item We first define the \textit{Rejection Markov chain}, a generalization of the rejection sampling in a context of MCMC sampling. This part allows us to define the rejection kernel used further on, and provides some intuition for the algorithm in the specific case where there exists a known constant $M$ such that the density ratio $\frac{\pi}{M\tilde\pi}$ is upper-bounded by $1$.
    \item We then generalize the Rejection Markov chain using repetitions to allow the density ratio $\frac{\pi}{M\tilde\pi}$ to be greater than $1$, thereby relieving the constraint of the first part. The idea is similar to IS as the number of repetitions is proportional to the density ratio, to the exceptions that: (1) the number of repetitions is a random integer and the constraint is simply that its expectation is proportional to the density ratio; (2) the output is a true random sample and not a weighted one as in classical IS. We use an extended space to construct an augmented Markov chain composed of repetitions of the instrumental chain as its first component and an integer as its second, keeping track of the remaining number of remaining repetitions. We then proceed to establish some theoretical properties, under mild assumptions, notably a law of large numbers, a Central Limit Theorem, a geometric ergodicity property and some uniqueness results.
    \item Finally, we illustrate the IMC on two synthetic examples. The first is a multidimensional mixture of Gaussian distributions, using as instrumental distribution a tempered version of the target, and a NUTS kernel. The second focuses on an i.i.d. sample from the instrumental distribution, defined by a normalizing flow trained to approximate a multimodal target up to dimension 25. 
\end{enumerate}

\section{Notations}

Let us denote $\rset^+ = [0,\infty)$, $\nset = \{1,2,...\}$ and $\nsetzero = \nset \cup \{0\}$.
We use the standard convention that $\prod_{k=m}^{n}=1$ if $m>n$. For
integers $k \leq \ell$, the notation $[k:\ell]$ stands for the set
$\{k,\ldots,\ell\}$ and in case where $k>\ell$, $[k:\ell]$ is the
empty set. Moreover, $u_{k:\ell}=(u_k,u_{k+1},\ldots,u_\ell)$ for all
$k\leq \ell$ and $u_{k:\infty}=(u_\ell)_{\ell \geq k}$. If a space
$\Xset$ is equipped with a $\sigma$-field $\Xsigma$, we denote by
$\posfunc{\Xset}$ the set of all nonnegative measurable functions with
respect to $\Xsigma$, that is, we make implicit the dependence on the
$\sigma$-field $\Xsigma$ in the notation $\posfunc{\Xset}$. Similarly $\funcset{\Xset}{b}$ is the set of all bounded measurable functions on $\Xset$ and $\funcset{\Xset}{b+}=\funcset{\Xset}{b} \cap \funcset{\Xset}{+}$. Moreover,
we denote by $\measureset_{1}(\Xset)$ the set of  probability measures
on $(\Xset,\Xsigma)$. For a non-negative real number $x$, we denote the
floor function by $\floor{x}$ and the fractional part by
$\fracpart{x}$ and hence $x = \floor{x} + \fracpart{x}$. The positive part of $x$ is written $(x)^+$.  

If $P,Q$ are Markov kernels on $\Xset \times \Xsigma$, $h$ a measurable function on $\Xset$, $\nu$ a measure on $(\Xset,\Xsigma)$ we define:

\begin{itemize}
    \item $Ph(x) := \int_\Xset P(x,\rmd y)h(y)$, for $x \in \Xset$,
    \item $\nu P(\ens{A}) := \int_\Xset \nu(\rmd x)P(x,\ens{A})$ for $A \in \Xsigma$,
    \item $PQ(x,\ens{A}) := \int_\Xset P(x,\rmd y)Q(y,\ens{A})$  for $x \in \Xset$ and $\ens{A} \in \Xsigma$,
    \item $\nu(h) := \int_\Xset h(x)\nu(\rmd x)$, also denoted $\nu h$ if the context is clear.    
\end{itemize}

Furthermore, we simply denote $P^0 := I$ and for $k \in  \nset$, $P^k := P P^{k-1} = P^{k-1}P$.
Last, a kernel $P$ is said \textit{sub-Markovian} if for all $x\in\Xset$, $P(x,\Xset)\leq 1$.

\section{The Rejection Markov chain}

Rejection sampling is a standard in Monte Carlo simulation. 
In this algorithm, we create samples distributed according to the target $\pi$ by subsampling among a batch of iid random variables distributed according to the instrumental $\tilde \pi$.

Our Markov rejection algorithm closely resembles the rejection algorithm except that we no longer need to subsample from an i.i.d. batch of random variables exactly distributed according to $\tilde \pi$, which can be restrictive. Instead, we rely on a Markov chain targeting
$\tilde \pi$, i.e. generated by a Markov kernel with invariant probability $\tilde \pi$, which can be achieved for example via a Metropolis Hastings algorithm. Note that it is sufficient to know $\tilde \pi$ up to a normalizing constant. Subsampling is done by accepting each candidate sample $X$ with probability $\rho(X)$ where $\rho:\Xset\to [0,1]$ is a well-chosen function. 

\subsection{Formal definition}

Let $(\Xset,\Xsigma)$ be a measurable space. For a given  Markov
kernel $Q$ on $\Xset \times \Xsigma$, we denote by $\PP_\xi^Q$ the
probability measure induced on $(\Xset^\nsetzero,\Xsigma^{\otimes \nsetzero})$
by the Markov kernel $Q$ and the initial distribution $\xi$, and by
$\PE^Q_\xi$ the associated expectation operator. If $\xi = \delta_x$
for some $x \in \Xset$, we simply use $\EE_x^Q :=
\EE^Q_{\delta_x}$. On $(\Xset^\nsetzero,\Xsigma^{\otimes \nsetzero})$, we define $X_\ell$ as the
projection on $\ell^{th}$-component, i.e., for any $w=(w_\ell)_{\ell \in \nsetzero}\in\Xset^\nsetzero$, $X_\ell(w)=w_\ell$, and $\theta$ the shift operator on $\Xset^\nsetzero$ such that $\theta:(w_0, w_1, ...)\mapsto (w_1, w_2, ...)$. For any measurable
function $\rho: \Xset \to [0,1]$, the rejection kernel
$\rej$ is defined  as follows:
\begin{equation} 
    \label{eq:def:rej}
    \rej h(x)= \sum_{k=1}^\infty  \PE_x^Q \lrb{h(X_k) \rho(X_k)   \prod_{i=1}^{k-1} (1-\rho(X_i))},
\end{equation}
where $x \in \Xset$ and $h$ is any nonnegative measurable function on
$(\Xset,\Xsigma)$. A transition according to $S$ is obtained by generating a Markov chain $\seqq{X_k}{k\in\nsetzero}$ according to the kernel $Q$ and by selecting the first accepted candidate among 
$\seqq{X_k}{k\in\nsetzero}$ with the success
probability sequence $\seqq{\rho(X_k)}{k\in\nsetzero}$. More precisely, define $\Yset=\Xset \times [0,1]$ and $\Ysigma= \Xsigma \otimes \mcbb([0,1])$ and let $G$ be the Markov kernel on $\Yset \times \Ysigma$ such that for all $y=(x,u) \in \Yset$ and all $\ens{A}\in\Ysigma$,
\begin{equation}
    G(y,\ens{A})=\int_{\Yset} \indi{\ens{A}}(x',u') Q(x,\rmd x')  \rmd u' \eqsp.
    \label{eq:kernelG}
\end{equation}
Therefore if $Y'=(X',U') \sim G(y,\cdot)$, then $X'$ and $U'$ are independent and marginally, $X'\sim Q(x,\cdot)$ and $U'\sim \mathcal{U}(0,1)$. For the Markov chain $\seqq{Y_k=(X_k,U_k)}{k\in \nsetzero}$ with Markov kernel $G$, define the first return time to the set $\ens D=\sett{y=(x,u)  \in \Yset}{u\leq \rho(x)}$ by 
\begin{equation}
        \sigma_{\ens D}=\inf \sett{k\geq 1}{Y_k \in \ens D}\eqsp.  
\end{equation}
Then, $Sh(x)=\PE_{\delta_x \otimes \gamma}^{G}[\indiacc{\sigma_{\ens D}<\infty}h(X_{\sigma_{\ens D}})]$ where $\gamma$ is any probability measure
on $[0,1]$, showing on the side that $S\indi{\Xset}(x)=\PP_{\delta_x \otimes \gamma}^{G}(\sigma_{\ens D}<\infty) \leq 1$, which implies that the kernel $S$ is sub-Markovian.

\begin{proposition}
    \label{prop:rejection}
    Let $Q$ be a Markov kernel on $\Xset \times \Xsigma$ with invariant
    probability measure $\piaux \in \measureset_1(\Xset)$ and let
    $\rho: \Xset \to [0,1]$ be a measurable function. Provided that $\mu(\rho)>0$, the probability measure $\nu$ defined by
    \begin{equation}
        \label{eq:def:nu:mu}
        \nu(\ens{A})=\frac{\int_{\ens{A}} \rho(x) \piaux(\rmd x)}{\int_\Xset \rho(x) \piaux(\rmd x)}\,, \quad \ens{A}\in\Xsigma\eqsp,
    \end{equation}
    is invariant with respect to $S$, i.e.  $\nu \rej=\nu$.
\end{proposition}

\begin{proof}
Define $\bar \rho\eqdef 1-\rho$. 
    For any bounded function $h\in \posfunc{\Xset}$ and any $x\in\Xset$,

    \begin{align}
        Sh(x)  &= \PE_x^Q \lrb{\sum_{k=1}^\infty \rho(X_k)  h(X_k)
          \prod_{i=1}^{k-1} \bar \rho(X_i)}    \nonumber \\
              & = \PE_x^Q \lrb{\rho(X_1)  h(X_1) }+  \PE_x^Q \lrb{ \bar \rho(X_1) \sum_{\ell=1}^\infty
                \rho(X_{\ell+1})  h(X_{\ell+1}) \prod_{j=1}^{\ell-1}
                \bar \rho(X_{j+1})} \nonumber \\
              & =Q(\rho h)(x)+Q(\bar \rho Sh)(x)\eqsp. \label{eq:S recursive}
    \end{align}
    Integrating with respect to $\mu$ yields
    \begin{equation*}
        \piaux Sh= \piaux Q(\rho h)+\piaux Q(\bar \rho Sh) = \piaux
        (\rho h)+\piaux (\bar \rho Sh)\eqsp,
    \end{equation*}
    where we used $\mu Q=\mu$ in the last equality. Since $h$ is
    bounded, $\piaux Sh<\infty$. Retrieving $\piaux Sh$ on
    both sides, we finally obtain $\piaux (\rho S h)= \piaux(\rho h)$. Hence $\nu(S h)=\nu(h)$.
\end{proof}

\subsection{Application to sampling}
Let $\pi \in \measureset_1(\Xset)$ be the \textit{target} distribution and
denote by $\tilde \pi \in \measureset_1(\Xset)$ an \textit{instrumental} distribution. As for rejection or importance sampling, the goal is to produce a sample targeting $\pi$ using a sample of $\tilde \pi$, here obtained by using a Markov kernel $Q$. We denote $\seqq{\tilde X_i}{i\in \nsetzero}$ a Markov chain with transition kernel $Q$ and make the following hypothesis on $Q$ :

\begin{hyp}{H}
    \item \label{hyp:Qinv} The  Markov kernel $Q$ admits $\tilde \pi$ as invariant probability measure.
\end{hyp}
We also need the following domination assumption, which is compulsory for rejection sampling.
\begin{hypn}{H$_{\mbox{\tiny rej}}$}
    \item \label{hyp:condition:rejet} There exists $M>0$ such that $\pi \leq M \tilde \pi.$
\end{hypn}
Then, we can use \eqref{eq:def:nu:mu} to define $\rho$ such that $\pi$ is the invariant probability measure for the Markov kernel $S$. Indeed, if $\mu = \tilde \pi$, we get $\nu = \pi$ in \eqref{eq:def:nu:mu} by defining 
\begin{equation*}
    \rho \propto \frac{\rmd \pi}{\rmd \tilde \pi} \eqsp.
\end{equation*}
If in addition, we want $\rho$ to take values in $[0,1]$, we may pick
\begin{equation}
    \label{eq:rho_rejection}
    \rho(x) = \frac{1}{M}\frac{\rmd \pi}{\rmd \tilde \pi}(x) \eqsp,
\end{equation}
for $\tilde \pi$-almost all $x\in \Xset$. From Proposition~\ref{prop:rejection}, we deduce immediately:
\begin{theorem}
    Assume \ref{hyp:Qinv} and \ref{hyp:condition:rejet} and take $\rho$ as defined in \eqref{eq:rho_rejection}. Then $S$ is $\pi$-invariant.
\end{theorem}

Note that standard rejection sampling consists in applying a transition according to $S$ to the particular case where  $Q(x,\cdot)=\tilde \pi(\cdot)$, $\mu=\tilde
    \pi$, $\rho$ as in (\ref{eq:rho_rejection}) and hence $\nu=\pi$.

\section{The Importance Markov chain}
As for classical rejection sampling, the Markov chain rejection sampling
suffers from the drawback that \ref{hyp:condition:rejet} may be
satisfied only with a prohibitively large $M$ (or worse,
\ref{hyp:condition:rejet} may not even be satisfied). The sampling is in that case
inefficient since the average acceptance ratio is equal to $1/M$.
Actually, $\rho$ can be interpreted in the following way: given $\tilde X_{k-1}$, we draw a new point $\tilde X_k \in \Xset$ according to $Q$ and insert $\tilde N_k \sim \text{Ber}(\rho(\tilde X_k))$ replica in current sample. 
Therefore $\rho(\tilde X_k)=\EE[\tilde N_k \vert \tilde X_k]$. This can be viewed as the equivalent of the weight of $\tilde X_k$ in a importance sampling context.

The idea with $\varrho: \Xset \to \rset^+$ is similar but we now
offer to replicate $\tilde X_k$ a random number of times $\tilde N_k$
where the conditional 
expectation of the random integer  $\tilde N_k \in\nsetzero$ w.r.t. $\tilde
X_k$ is equal to $\varrho(\tilde X_k)$, as sketched in \Cref{algo:imc0}. This relates to \cite{gasemyrMarkovChainMonte2002} where the author conditions the weights of his estimate to be nonnegative integers.

\subsection{The extended Markov chain}

Let $\pi$ and $\tilde \pi$ be two probability measures on $(\Xset
    ,\Xsigma)$ and let $\kappa$ be a positive real number. Assume
that $\pi$ is dominated by $\tilde \pi$ and let $\varrho_\kappa: \Xset \to
    \rset^+$ be a measurable function such that
\begin{equation}
    \label{eq:def rho kappa}
    \varrho_\kappa (x)=\kappa \frac{\rmd \pi}{\rmd \tilde \pi} (x) \eqsp,
\end{equation}
for $\tilde \pi$-almost all $x\in\Xset$. 
For $\kappa=1$ let us simply denote
\begin{equation}
    \label{eq:def:rho:1}
    \varrho := \varrho_1 .
\end{equation} Let $\tilde R$ be a Markov kernel
on $\Xset \times  {\mathcal P} (\nsetzero)$ and by an abuse of notation, let us write 
$\tilde
    R(x,n)=\tilde R
    (x,\{n\})$ 
    for any $(x,n)\in\Xset \times  \nsetzero$. The distribution
$\tilde R(\Xaux_k,\cdot)$ will be used to draw the number $\tilde N_k$ of
replications of $\Xaux_k$ under the \textit{unbiasedness assumption:}
\begin{hyp}{H}
    \item  \label{hyp:unbiased} For all $x\in\Xset$,
    $$
        \sum_{n=0}^\infty \tilde R(x,n)n= \varrho_\kappa (x)\eqsp.
    $$

\end{hyp}
In this paper, we will face two cases : (1) $\tilde R$ and $\varrho_\kappa$ are available in closed form up to a normalizing constant, (2) $\tilde R$ can be generated under \ref{hyp:unbiased} without the explicit knowledge of $\varrho_\kappa$. Furthermore, $\kappa$ can be chosen arbitrarily even though we will discuss some interesting choices later on (see \Cref{subsec:optimal} and \Cref{subsec:choice:kappa}).
\begin{algorithm}[h!]
    \caption{}
    \label{algo:imc0}
    \begin{algorithmic}[1]
        \State $X=[\ ]$
        \State Set an arbitrary $\tilde X_0$.
        \For{$k\gets 1$ to $n$}
        \State Draw $\Xaux_k \sim Q (\Xaux_{k-1}, \cdot)$ and $\tilde N_k \sim \tilde R(\Xaux_k,\cdot)$
        \State Append $\tilde N_k$ replicas of $\Xaux_k$ to $X$
        \EndFor
        \State \textbf{output:} $X$
    \end{algorithmic}
\end{algorithm}

As seen below, if the Markov kernel $Q$ is $\tilde \pi$-invariant, then the output sequence $X=(X_0,X_1,\ldots)$ of Algorithm~\ref{algo:imc0} targets $\pi$. However, $X$ is not a Markov chain {\it per se}, and in order to study its ergodic properties, we need add a second component $N$ to the sequence $X$ so that the augmented sequence $(X,N)$ becomes a Markov chain. This is done by rewriting Algorithm~\ref{algo:imc0} as Algorithm~\ref{algo:imc}.

\begin{algorithm} [h!]
    \caption{Importance Markov chain (IMC)}
    \label{algo:imc}
    \begin{algorithmic}[1]
        \State $\ell \gets 0$
        \State Set an arbitrary $\tilde X_0$.
        \For{$k\gets 1$ to $n$}
        \State Draw $\Xaux_k \sim Q (\Xaux_{k-1}, \cdot)$ and $\tilde N_k \sim \tilde R(\Xaux_k,\cdot)$
        \State Set $N_\ell=\tilde N_k$
        \While{$N_\ell \geq 1$}
        \State Set $(X_\ell,N_\ell)\gets (\Xaux_k,N_\ell-1)$
        \State Set $\ell \gets \ell+1$
        \EndWhile
        \EndFor
    \end{algorithmic}
\end{algorithm}

Let us describe the transition of the extended Markov chain $\seqq{(X_\ell,N_\ell)}{\ell  \in \nsetzero}$.
From Algorithm~\ref{algo:imc}, we can see that $(X_\ell,N_\ell)$ is updated according to two different moves. Either we are already inside the while loop described in lines 6-9 of Algorithm~\ref{algo:imc}, in which case, $N_\ell \geq 1$ and the update is simply $(X_{\ell+1},N_{\ell+1})=(X_\ell,N_\ell-1)$, or we are outside the while loop, in which case, $N_\ell=0$ and $X_\ell=\tilde X_k$ for some $k\in \nsetzero$. Then, the update of $(X_\ell,N_\ell)$ happens when we enter again the while loop, in which case $(X_{\ell+1},N_{\ell+1})=(\tilde X_T,\tilde N_T-1)$ where $T=\inf\sett{n >k}{\tilde N_n\neq 0}$.

The associated Markov kernel $P$ on $(\Xset \times
    \nsetzero) \times (\Xsigma \otimes \mathcal{P}(\nsetzero))$ is then defined by: for all $h \in \posfunc{\Xset \times \nsetzero}$,
\begin{equation} 
    \label{eq:def:P:Q-Rtilde}
    Ph(x,n)=\indiacc{n \geq 1} h(x,n-1) + \indiacc{n=0} \cdot \sum_{k,\ell=1}^\infty  \PE_x^Q \lrb{h(X_{k},\ell-1) \tilde R(X_{k},\ell) \prod_{i=1}^{k-1} \tilde R(X_i,0)}.
\end{equation}
To simplify this expression, let us introduce additional notation. Write
$\rho_{\tilde R}(x) =\tilde R (x,[1:\infty))$ and let $S$ be the Markov kernel on $\Xset \times  \Xsigma$ defined by
\begin{equation} \label{eq:def:S}
        Sf(x)=\sum_{k=1}^\infty \PE_x^Q\lrb{
        f(X_k)\rho_{\tilde R}(X_k)\prod_{i=1}^{k-1}(1-\rho_{\tilde R}(X_i))}\, ,
\end{equation}
for $ f\in \posfunc{\Xset} .$

The kernel $S$ is of the same form as in \eqref{eq:def:rej} except that
$\rho$ is now replaced by $\rho_{\tilde R}$. Note that by construction, $\rho_{\tilde R}$ is $[0,1]$-valued and $\rho_{\tilde R} (x)$ can be interpreted as the probability for an $\tilde X_i=x$ drawn from $Q$ to be accepted for the chain $(X_\ell)$, in which case, we keep at least one replica of $\tilde X_i$.

Then, the extended Markov kernel $P$ writes, for $h \in \posfunc{\Xset \times \nsetzero}$:
\begin{equation}
        \label{eq:def:P}
        Ph(x,n)=\indiacc{n \geq 1} h(x,n-1)  +\indiacc{n=0} \sum_{n'=0}^\infty  \int_{\Xset}  S(x,\rmd x')  R(x',n')h(x',n') \eqsp,
\end{equation}
where $R$ is the Markov kernel on $\Xset\times {\mathcal P} (\nsetzero)$
defined by
\begin{equation}
    \label{eq:R}
    R(x,n)\eqdef  \tilde
    R(x,n+1)/\rho_{\tilde R} (x)\,, \quad (x,n)\in \Xset \times
    \nsetzero\eqsp,
\end{equation}
and where we, again, make the abuse of notation $R(x,n)\eqdef
    R(x,\{n\})$.

Note that \eqref{eq:def:P:Q-Rtilde} and \eqref{eq:def:P} give two
different but equivalent decompositions of 
$P$, for the sampling step (when $n=0$). In \eqref{eq:def:P:Q-Rtilde}, we sample $\tilde
X_i$ according to $Q$ and then use $\tilde R(\tilde X_i,\cdot)$ to
draw a number of replicas $\tilde N_i$, until it is larger than $1$, in which case we retain $\tilde X_i$ and $\tilde N_i-1$, the number of remaining replicas. 
In \eqref{eq:def:P}, we bypass the rejection step by drawing directly a new accepted point $X_i$ using $S$ and then the number of remaining replicas from $R(X_i,\cdot)$, which corresponds to the law of $\tilde N-1$ conditionnally on $\{\tilde N\geq 1\}$ when $\tilde N\sim \tilde R(X_i,\cdot)$.

\begin{remark}
    The unbiasedness assumption \ref{hyp:unbiased} is closely related to the notion of \textit{correctly weighted} density developed in \cite{wongDynamicWeightingMonte1997,liuTheoryDynamicWeighting2001}. Write $\hat \pi(\dd x\dd n) = \tilde \pi (\dd x) \tilde R(x,\dd n)$ a joint distribution on  $\Xset \times \nsetzero$. Then, under \ref{hyp:unbiased}, $\sum_{n \in \NN} \tilde \pi (\dd x) \tilde R(x,n)n  = \kappa \pi(\dd x)$ so $\hat  \pi$ is \textit{correctly weighted}. And by construction, the kernel $Q(x,\dd y) \tilde R(y,dn)$ that generates the samples $(\tilde X_i, \tilde N_i)_{i \in \NN}$ of \Cref{algo:imc0} admits $\hat \pi$ as an invariant probability distribution (see \Cref{lemma:invMeas:Qbar}).
\end{remark}

\begin{remark}
    While importance sampling requires exact simulations from $\tilde \pi$, the IMC method only relies on a Markov kernel $Q$ targetting $\tilde \pi$. This allows us to extend the set of usable instrumental distributions. 
\end{remark}

\begin{remark} \label{remark:MH}
    Perhaps surprisingly, Metropolis-Hastings (MH) algorithms
    can be cast into the framework of importance Markov
    chains. Indeed, take a MH algorithm with proposition kernel
    $A(x,dy)$  and acceptance rate $\alpha(x,y)$, targeting
    $\pi$. Following the framework of \cite{doucVanillaRaoBlackwellization2011}, the accepted points $\seqq{\tilde X_i}{i \in\nsetzero}$ form a Markov chain with Markov kernel $Q(x,dy)$ proportional to $\alpha(x,y) A(x,dy)$. Before moving to a new accepted point, $\tilde X_i =x$ is repeated a random number of times $\tilde N_i$ that follows (conditionally on $\tilde X_i =x$) a geometric distribution with success probability $p(x) := \int_\Xset \alpha(x,y)A(x,dy)$. Then it can be shown that $Q$ is $\tilde \pi$-invariant where $\tilde \pi(dx) \propto p(x)\pi(dx)$, hence \ref{hyp:Qinv} holds. Define $\tilde R(x,\cdot)$ as the geometric distribution with parameter $p(x)$. Choosing $\kappa = 1/ \int_\Xset p(x)\pi(\dd x)$, $\varrho_\kappa$ defined in \eqref{eq:def rho kappa} writes $\varrho_\kappa(x) = 1/p(x)$, hence \ref{hyp:unbiased} holds. Then, with these choices of $Q$ and $\tilde R$, $(\tilde X_i,\tilde N_i)$ corresponds to the output of the IMC algorithm defined in \Cref{algo:imc0}.   

\end{remark}

\subsection{Invariant probability measure}

\subsubsection{Existence}
Let $\bar \pi$ be the measure on $\Xset \times
    \nsetzero$ defined by: for any $h \in \posfunc{\Xset \times \nsetzero}$,
\begin{align}
    \label{eq:def:pibar}
    \bar \pi(h) & =\kappa^{-1} \sum_{n=1}^\infty \int_{\Xset}  \tilde \pi(\rmd x)
    \tilde R(x,n)\sum_{k=0}^{n-1} h(x,k) \nonumber                                         \\
                      & =\kappa^{-1} \sum_{\ell=0}^\infty \int_{\Xset}  \tilde \pi(\rmd x)
    \rho_{\tilde R} (x) R(x,\ell)\sum_{k=0}^{\ell} h(x,k)\eqsp,
\end{align}
where the last equality follows from \eqref{eq:R} and the change of
variable $\ell=n-1$.

\begin{proposition}
    \label{prop:invariance:imc}
    Assume \ref{hyp:Qinv} and \ref{hyp:unbiased}. 
    Let $P$ be the Markov kernel defined in \eqref{eq:def:P} and let $\bar \pi$ be the probability measure on $\Xset \times \nsetzero$ defined in \eqref{eq:def:pibar}. Then, 
    \begin{enumerate}[label=(\roman*)]
        \item \label{item:invarince:imc:two} the Markov kernel $P$ is $\bar \pi$-invariant,
        \item \label{item:invarince:imc:one} the marginal of $\bar \pi$ on the first component is
        $\pi$.
    \end{enumerate}
\end{proposition}
\begin{proof}
    We start with \ref{item:invarince:imc:two}. Let $h \in \posfunc{\Xset \times \nsetzero}$. Interchanging the sum in $
        \ell$ and the sum in $k$ in (\ref{eq:def:pibar}) yields
    \begin{equation}
        \label{eq:def;new:pibar}
        \bar \pi (h) =\kappa^{-1} \sum_{k=0}^{\infty} \int_{\Xset}
        \tilde \pi(\rmd x)\rho_{\tilde R} (x) R(x,[k:\infty)) h(x,k)
        \eqsp.
    \end{equation}
    We now replace $h$ by $Ph$ and combine with the expression of $Ph$ given in (\ref{eq:def:P}), we then obtain

    \begin{align*}
        \bar \pi (Ph)  &=\kappa^{-1} \sum_{k=1}^{\infty} \int_{\Xset}
        \tilde \pi(\rmd x) \rho_{\tilde R} (x) R(x,[k:\infty) )
        h(x,k-1)  +\kappa^{-1} \sum_{n'=0}^\infty  \int_{\Xset}
        \tilde \pi(\rmd x)\rho_{\tilde R} (x) \int_{\Xset} S(x,\rmd x')
        R(x',n')h(x',n')\\
        & =\kappa^{-1} \sum_{n'=0}^{\infty} \int_{\Xset}
        \tilde \pi(\rmd x)\rho_{\tilde R} (x) R(x,[n'+1:\infty))
        h(x,n')  +\kappa^{-1} \sum_{n'=0}^\infty  \int_{\Xset}
        \tilde \pi(\rmd x)\rho_{\tilde R} (x)
        R(x,n')h(x,n')\eqsp,
    \end{align*}
    where the last equality follows (a) from the change of variable
    $n'=k-1$ for the first term of the rhs and (b) from Proposition~\ref{prop:rejection} applied, under \ref{hyp:Qinv}, to $\rho=\rho_{\tilde R}$
and $\piaux=\tilde \pi$ for the second term. 
Noting that $R(x,[n'+1:\infty)) +R(x,n')=R(x,[n':\infty))$, we finally get
\begin{equation}
    \begin{split}
        \bar \pi (Ph)&=\kappa^{-1} \sum_{n'=0}^{\infty} \int_{\Xset} \tilde \pi(\rmd x)\rho_{\tilde R} (x) R(x,[n':\infty)) h(x,n') =\bar \pi (h) \eqsp,
    \end{split}
\end{equation}
where (\ref{eq:def;new:pibar}) is used to obtain the last equality.

We now turn to \ref{item:invarince:imc:one}. For any $\ens{A}\in \Xsigma$, applying (\ref{eq:def:pibar})
with $h(x,k)=\indi{\ens{A}} (x)$ yields under \ref{hyp:unbiased}
\begin{equation}
    \begin{split}
        \bar \pi(\ens{A} \times \nsetzero)&=\kappa^{-1}\int_\Xset \tilde \pi(\rmd
        x)
        \lr{\sum_{n=1}^\infty \tilde R(x,n) n} \indi{\ens{A}}(x)= \kappa^{-1} \tilde \pi
        (\varrho_\kappa \indi{\ens{A}})=\pi (\ens{A}) \eqsp.
    \end{split}
\end{equation}
\end{proof}

\subsection{Uniqueness}

\begin{proposition} \label{prop:uniqueness}
    Assume \ref{hyp:Qinv} and \ref{hyp:unbiased}. If any invariant
    measure for $Q$ is proportional to $\tilde \pi$ (defined in \ref{hyp:Qinv}), then $\bar \pi$ defined in \eqref{eq:def:pibar} is the unique invariant probability measure for $P$. 
\end{proposition}
\begin{proof}
See \ref{secA1}.
\end{proof}
The uniqueness of the invariant probability measure $\bar \pi$ for $P$, as stated in \Cref{prop:uniqueness} allows to obtain the Birkhoff ergodic theorem (\cite[Theorem 5.2.9]{doucMarkovChains2018}): for any measurable function $g: \Xset \times \nsetzero \rightarrow \rset$ such that $ \bar \pi( \abs{g}) < \infty$,
    \begin{equation}
        \lim_{n \rightarrow \infty} n^{-1} \sum_{k=0}^{n-1} g(X_k,N_k) =  \bar \pi(g), \quad \PP_{\bar \pi}^P-a.s.
    \end{equation}
Although reassuring, the law of large numbers holds $\PP_{\bar \pi}^P-a.s.$, i.e. the initial distribution is set to be the invariant probability measure $\bar \pi$, which is not realistic from a practical point of view. Consequently, we will now turn to conditions under which the law of large numbers holds, irrespective to the initial distribution.   
\subsection{Law of large numbers}
To establish a strong law of large numbers for the kernel $P$, we rely on the single hypothesis that the instrumental kernel $Q$ satisfies a law of large numbers. More precisely if the instrumental kernel $Q$ satisfies a law of large numbers for any initial distribution $\xi\in\measureset_1(\Xset)$, Theorem \ref{thm: SLLN of P} will show that it is also the case for the importance Markov kernel $P$.

\begin{hypn}{H$_{\mbox{\tiny lln}}$}

    \item For every $\xi \in \measureset_1(\Xset)$ and measurable function $g: \Xset \rightarrow \rset$ such that $\tilde \pi( \abs{g}) < \infty$,
    \begin{equation*}
        \lim_{n \rightarrow \infty} n^{-1} \sum_{k=0}^{n-1} g(\tilde X_k) = \tilde \pi(g), \quad \PP_\xi^Q-a.s.
    \end{equation*} \label{hyp:SLLN of Q}
\end{hypn}

\begin{theorem}
    \label{thm: SLLN of P}
    Assume \ref{hyp:Qinv} and \ref{hyp:SLLN of Q}. Then,
    for every $\xi \in \measureset_1(\Xset \times \nsetzero)$ and measurable function $g: \Xset \times \nsetzero \rightarrow \rset$ such that $ \bar \pi( \abs{g}) < \infty$,
    \begin{equation}
        \lim_{n \rightarrow \infty} n^{-1} \sum_{k=0}^{n-1} g(X_k,N_k) =  \bar \pi(g), \quad \PP_\xi^P-a.s.
    \end{equation}
\end{theorem}

\begin{proof}
    The proof relies on \cite[Proposition 3.5]{doucBoostYourFavorite2022}, which relates \ref{hyp:SLLN of Q} to a property on the harmonic functions for $Q$ (i.e. measurable functions $h$ such that $Qh=h$). More precisely, it states that for any Markov kernel $Q$ satisfying  $\tilde \pi Q = \tilde \pi$ for some $\tilde \pi \in \measureset_1(\Xset)$, \ref{hyp:SLLN of Q} is equivalent to \ref{hyp: harmonic function} defined as follows:

    \begin{hypn}{H$_{\mbox{\tiny hrm}}$}
        \item Any bounded harmonic function $h: \Xset \rightarrow \rset$ for $Q$  is constant. \label{hyp: harmonic function}
    \end{hypn}
    Hence, proving Theorem \ref{thm: SLLN of P} is equivalent to showing that any bounded harmonic function for $P$ is constant.

    Let $ \bar h : \Xset\times\nsetzero\rightarrow\rset$ be a bounded harmonic function for $P$. Then for $n>0$, 
    \begin{equation}
        \bar h(x,n) = P \bar h(x,n) = \bar h(x,n-1),
    \end{equation}
    where the last equality comes from \eqref{eq:def:P}. 
    Thus $\bar h$ does not depend on its second argument and we can define a measurable function $h:\Xset\rightarrow\rset$ such that 
    \begin{equation} \label{eq:def:h}
        h(x) = \bar h (x,n)
    \end{equation}
    for all $(x,n)\in\Xset\times\nsetzero$.
    Now,
   \begin{equation*}
        h(x) = \bar h(x,0) = P\bar h(x,0) = \int_{\Xset \times \nsetzero} S(x,dx') R(x',dn) \bar h(x,n) =Sh(x)
   \end{equation*}
   using the expression of $P$ in  \eqref{eq:def:P} as well as \eqref{eq:def:h}.
    From the recursive expression for $S$ in \eqref{eq:S recursive} we have:
    \begin{align*}
        h(x) & = Sh(x)                            \\
             & = Q(\rho h)(x) + Q((1-\rho)Sh)(x)  \\
             & = Q( \rho h)(x) + Q((1- \rho)h)(x) \\
             & = Qh(x).
    \end{align*}
    Therefore $h$ is harmonic for $Q$, and since it is also bounded, \ref{hyp:SLLN of Q} implies that it is constant.
    Then \eqref{eq:def:h} shows that $\bar h$ is constant which concludes the proof. 
\end{proof}

\subsection{Central Limit Theorem }

Let us now establish a Central Limit Theorem (CLT) associated to the Importance Markov chain for a particular function $h$, based on a similar hypothesis on the instrumental kernel $Q$ with the function $\varrho h$. More precisely, the Central Limit Theorem for $Q$ stems from the existence of a solution to the Poisson equation for this kernel. This is a quite common sufficient condition for a CLT, and although other conditions exist (references are given for example in \cite[Chap 21]{doucMarkovChains2018}), we choose this one for its simplicity in the proofs. To be more specific, we are interested in measurable functions $h: \Xset \to \rset$ such that if we define 
\begin{equation} \label{eq:h_0}
    h_0 := h- \pi h,
\end{equation}
the following condition holds:
\begin{hypn}{H$_{\mbox{\tiny Poiss}}$}
    \item \label{hyp:CLT}  The Poisson equation associated to $\varrho h_0$ for the kernel $Q$ on $\Xset$ admits a $\tilde \pi$-square integrable solution $H$, i.e. for all $x\in\Xset$,
    \[H(x) - QH(x) = \varrho h_0(x) \text { and } \tilde\pi H^2<\infty.\]
    In addition, $\int_{\Xset \times \nsetzero} n^2h(x)^2\tilde\pi(\rmd x)\tilde R(x,\rmd n)<\infty$.
\end{hypn}
\begin{remark}
    Note that $\tilde\pi(\varrho h_0) = 0$ since $\varrho = \frac{\rmd\pi}{\rmd\tilde\pi}$, hence this term does not appear in the Poisson equation.
\end{remark}
Under \ref{hyp:CLT}, \cite[Theorem 21.2.5]{doucMarkovChains2018} ensures that the Markov chain $(\tilde X_i)$ generated by the kernel $Q$ satisfies a Central Limit Theorem for the function $\varrho h_0$ :
\begin{equation}\label{eq:CLTstationary}
    \frac{1}{\sqrt{n}}\sum_{i=1}^{n}  \varrho(\tilde X_i) h_0(\tilde X_i) \convlaw{\Prob^Q_{\tilde\pi}} \mathcal{N}(0,\tilde \sigma^2(\varrho h_0)),
\end{equation}
where $ \tilde \sigma^2(\varrho h_0)= 2\tilde\pi\lr{\varrho h_0H} - \tilde\pi((\varrho h_0)^2)$. 
\Cref{lem:martingaleTCL} of \ref{sec:clt} combined with \ref{hyp:SLLN of Q} then extends the weak convergence under ${\Prob^Q_{\tilde\pi}}$ in the equation above to a weak convergence under ${\Prob^Q_{\xi}}$ for any $\xi\in\measureset_1(\Xset)$. 
We can now state the CLT for the Importance Markov chain in a formal manner.
\begin{theorem}
    \label{thm:CLT}
    Assume  \ref{hyp:Qinv}, \ref{hyp:unbiased}, \ref{hyp:SLLN of Q} and let $h: \Xset \to \rset$  be  a measurable function that satisfies \ref{hyp:CLT}. Then there exists a constant $\sigma^2(h) >0$ such that
        \begin{equation*}
            \frac{1}{\sqrt{n}}\sum_{i=1}^{n}  \lr{h(X_i) - \pi h} \convlaw{\Prob^P_{\chi}} \mathcal{N}(0,\sigma^2(h)),
        \end{equation*}
    where the distribution $\chi$ is defined by $\chi(f)=\int\xi(\rmd x) S(x,\rmd x')R(x',\rmd n')f(x',n')$. 
    Moreover, we have the following expression of $\sigma^2(h)$:
        \begin{equation} \label{eq:clt:var}
            \sigma^2(h) =   \kappa  \tilde \sigma^2(\varrho h_0) + \kappa^{-1} \hat \sigma^2(h_0,\kappa),
        \end{equation}
where \begin{description}
    \item [-] $\tilde \sigma^2(\varrho h_0)$ is the variance obtained in \eqref{eq:CLTstationary},
    \item  [-] $\hat \sigma^2(h_0,\kappa) := \int_\Xset h_0^2(x)\Var^{\tilde R}_x[N]\tilde \pi(dx)$,
    \item [-] $\Var^{\tilde R}_x[N] := \int_\nsetzero\tilde R(x,\rmd n)n^2 - \lr{\int_\nsetzero\tilde R(x,\rmd n)n}^2$.
\end{description}
\end{theorem}
\begin{proof}
    See \ref{sec:clt}.
\end{proof}

\begin{remark}
    \label{rmk:choice of R}
    Note that the variance $\sigma^2(h)$ can be decomposed into two terms: $(1)$ $\tilde \sigma^2(\varrho h)$ is the variance coming from the instrumental chain, while $(2)$ $\hat \sigma^2 (h,\kappa)$ is the variance brought by the random number of repetitions of the instrumental chain.
\end{remark}    
    
\subsection{Minimizing the asymptotic variance}

\subsubsection{Optimal choice of the kernel $\tilde R$}
Following Remark \ref{rmk:choice of R}, one can notice that the expression $\hat \sigma ^2(h,\kappa) = \int_\Xset h^2(x)\Var^{\tilde R}_x[N]\tilde \pi(dx)$ directly depends on the variance of $N$ under $\tilde R(x,\cdot)$. 
Therefore, minimizing the variance associated to $\tilde R(x,\cdot)$, for $x \in \Xset$, leads to minimization of the asymptotic variance of the chain as defined in Theorem \ref{thm:CLT}. To help tuning $\tilde R$, we state the following lemma:

\begin{lemma}
    \label{lem:var:min:int:RV}
    Let $N$ be an integer-value random variable on some probability space $(\Omega, \mathcal{F}, \mathbb{P})$ such that $\EE[N] = \rho$ for a fixed $\rho \in \rset^+$. Then,
    
    \begin{equation}
        \Var(N) \geq \fracpart{\rho}(1 - \fracpart{\rho}) . 
    \end{equation}

    This bound is reached for $N = \floor{\rho} + S$, where $S \sim Ber(\fracpart{\rho})$
\end{lemma}

\begin{proof}
    Using $ 0 = \EE [N] - \rho = \EE[(N-\rho)^+] - \EE[(N-\rho)^-] $
    and $ (N-\rho)^2 \geq (1- \fracpart{\rho}) (N-\rho)^+ + \fracpart{\rho}(N-\rho)^- $,
    we get 
    \begin{equation*}
        \EE [(N-\rho)^2] \geq  \EE[(N-\rho)^+] = \EE[(N- \rho)^-] .
    \end{equation*}

    \begin{itemize}
        \item If $\PP(N> \rho) \geq \fracpart{\rho}$, then 
        \begin{equation*}
            \EE [(N-\rho)^2] \geq  \EE[(N-\rho)^+]\geq (1- \fracpart{\rho})P(N>\rho)\geq \fracpart{\rho}(1 - \fracpart{\rho}).
        \end{equation*}
        \item If $\PP(N> \rho) < \fracpart{\rho}\Leftrightarrow\PP (N \leq \rho ) > 1- \fracpart{\rho}$, then
        \begin{equation*}
            \EE [(N-\rho)^2] \geq  \EE[(N-\rho)^-]\geq \fracpart{\rho}P(N \leq \rho)\geq \fracpart{\rho}(1 - \fracpart{\rho}).
        \end{equation*}
    \end{itemize}
\end{proof}    
In the case where $\varrho_\kappa$ can be computed, we can use \Cref{lem:var:min:int:RV} to define  $\tilde R$ as:
\begin{equation}
    \Ropttilde = (1-\fracpart{\varrho_\kappa (x)}) \delta_{\floor{\varrho_\kappa(x)}} + \fracpart{\varrho_\kappa(x)} \delta_{\floor{\varrho_\kappa(x)}+1}.
    \label{eq:optimal:tildeR}
\end{equation}
This $\Ropttilde$ implies the following expression for $\Ropt$:
\begin{equation*}
    \Ropt= (1 - \fracpart{\varrho_\kappa(x)})\delta_{(\floor{\varrho_\kappa(x)}-1)^+} + \fracpart{\varrho_\kappa(x)} \delta_{\floor{\varrho_\kappa(x)}} .
\end{equation*}

\subsubsection{Optimal upper bound}\label{subsec:optimal}
With the optimal choice of $\tilde R$ given in \eqref{eq:optimal:tildeR}, we have $\Var^{\Ropttilde}_x[N]=\fracpart{\varrho_\kappa (x)}(1-\fracpart{\varrho_\kappa (x)}) \leq 1/4$. Hence, 
\begin{align*}
    \sigma^2(h) &=   \kappa  \tilde \sigma^2(\varrho h_0) + \kappa^{-1} \hat \sigma^2(h_0,\kappa)\\
    & \leq  \kappa  \tilde \sigma^2(\varrho h_0) + \kappa^{-1} \tilde \pi(h_0^2)/4.    
\end{align*}
Therefore, for a given function $h$, optimizing the rhs of the above inequality yields: 
\begin{equation} \label{eq:kappa:opt}
    \kappa= \frac 1 2 \sqrt{\frac{\tilde \pi(h_0^2)}{\tilde \sigma^2(\varrho h_0)}} 
\end{equation}
from which we deduce the upper bound 
\[\sigma^2(h)  \leq \sqrt{\tilde \pi(h_0^2) \tilde \sigma^2(\varrho h_0)}\,.\]
Note that the choice of $\kappa$ in \eqref{eq:kappa:opt} depends on the function $h$ which is usually not given beforehand in practice. We propose another way of choosing $\kappa$ in \Cref{subsec:choice:kappa}, more adapted to practical concerns.

\subsection{Geometric ergodicity}

For any set $\ens{A}\in \Xsigma$, we use the notation $\bar{\ens{A}}=\Xset \setminus \ens{A}$ to denote its complement. Define the set $\ens{C}_\eta \eqdef \sett{x \in\Xset}{\rho_{\tilde R}(x) \geq \eta}$ for $\eta\in [0,1]$ and note that
\begin{equation}
    \label{eq:cns}
    % à vérifier
    x \in \ens{C}_\eta \Longleftrightarrow  1- \rho_{\tilde R}(x) \leq  1 - \eta \eqsp.
\end{equation}
Finally we denote by $\sigma_\ens{C} = \inf \sett{k\geq 1}{X_k \in \ens{C}}$ the first return time of the set $\ens{C}$.

\begin{lemma}
    \label{lem:smallset}
    Assume that for some $\eta\in(0,1)$, $\ens{C}_\eta$ is a $(1,\varepsilon \nu)$-small set for the kernel $Q$.
    Then, there exists a probability measure $\tilde \nu$ on $\Xset \times \nsetzero$ satisfying
    \begin{enumerate}[label=(\roman*)]
        \item \label{item:smallset:exist} $ \ens{C}_\eta \times \{0\}$ is a $(1,\varepsilon \tilde \nu)$-small set for the kernel $P$.
        \item \label{item:smallset:pos} if $\nu (\ens{C}_\eta\cap\{\tilde R(.,1)>0\}) >0$,
        then 
        \[\tilde \nu (\ens{C}_\eta \times \{0\}) >0.\]
    \end{enumerate}    
\end{lemma}

\begin{proof}
    See \ref{appendix:proof:lem:2and3}.
\end{proof}

We now introduce the following assumption:

\begin{hyp}{H}
    \item \label{hyp:majoration N} There exists $\beta_0>1$ such that 
    \[\sup_{x\in \Xset} \int_\nsetzero\beta_0^n \tilde R(x,\rmd n) < \infty.\]
\end{hyp}

\ref{hyp:majoration N} is a relatively weak condition, and it is in fact necessary for geometric ergodicity of a Metropolis-Hastings algorithm, seen as an instance of the IMC algorithm (see \Cref{remark:MH}). 
Indeed, in this case, $\tilde R(x,\cdot)$ is a geometric distribution with success parameter $p(x)$ and therefore \ref{hyp:majoration N} is equivalent to $p$ being lower bounded away from zero. \cite[Proposition 5.1]{robertsTweedie1996} proves that this last condition is a necessary condition for geometric ergodicity of a Metropolis-Hastings algorithm  .

Also, \ref{hyp:majoration N} holds whenever the support of $\tilde R(x,\cdot)$ is uniformly bounded on
$\Xset$ and in particular when $\frac{\rmd \pi}{\rmd \tilde \pi}$ is upper-bounded, and $\tilde R(x,\cdot)$ is the distribution of $\lfloor \varrho_\kappa(x) \rfloor+ U$ with $U \sim \mathrm{Ber}(\langle \varrho_\kappa(x)\rangle)$.

\begin{remark}
    Actually, although condition \ref{item:smallset:pos} of Lemma \ref{lem:smallset} is not verified for any kernel $\tilde R$, it is always possible to transform it slightly into a new kernel $\tilde R'$ satisfying this condition while keeping its other properties untouched, namely assumptions \ref{hyp:unbiased} and \ref{hyp:majoration N}. 
    Indeed, define 
    $\ens{C}_\eta^-=\{x\in \ens{C}_\eta : \tilde R(x,0)>0\}$ 
    and observe that $\nu (\ens{C}_\eta\cap\{\tilde R(.,1)>0\})>\nu (\ens{C}_\eta^-\cap\{\tilde R(.,1)>0\})$ since $\ens{C}_\eta^-\subset \ens{C}_\eta$. 
    We will now construct a kernel $\tilde R'$ such that for all $x\in \ens{C}_\eta^-$, $\tilde R'(x,1)>0$, which satisfies the desired condition as $\nu (\ens{C}_\eta^-\cap\{\tilde R'(.,1)>0\}) = \nu(\ens{C}_\eta^-)>0$. Let $x \in \ens{C}_\eta^-$, implying that $\tilde R(x,0)>0$. 
    Due to the unbiasedness assumption, there exists $k\in\nsetzero, k>1$ such that $\tilde R(x,k)>0$. Now define $\tilde R'$ such that: 

    \begin{itemize}
        \item $\tilde R'(x,n) = \tilde R(x,n)$ if $n\notin \{0, 1, k\}$,
        \item $\tilde R'(x,1) = \epsilon\tilde R(x,0)$,
        \item $\tilde R'(x,0) = \tilde R(x,0) - \frac{(k-1)\epsilon}{k}\tilde R(x,0)$,
        \item $\tilde R'(x,k) = \tilde R(x,k) - \frac{\epsilon}{k} \tilde R(x,0).$
    \end{itemize}
    Note that $\epsilon$ can be chosen small enough to guarantee $\tilde R'(x,k)>0$. One can easily check that
    \begin{itemize}
        \item $\sum_{n\geq 0}\tilde R'(x,n)=\sum_{n\geq 0}\tilde R(x,n) = 1$,
        \item $\sum_{n\geq 0}\tilde R'(x,n)n = \sum_{n\geq 0}\tilde R(x,n)n = \varrho_{\kappa}(x)$.
    \end{itemize}
    Hence \ref{hyp:unbiased} and \ref{hyp:majoration N} hold for $\tilde R'$.
\end{remark}

Recall  that $\ens{A} \in \Xsigma$ is {\em accessible} for the kernel $Q$ if and only if for all $x
\in \Xset$, there exists $n \in \nsetzero$ such that $Q^n(x,\ens{A})>0$ (see for example Lemma 3.5.2 of \cite{doucMarkovChains2018}). We then have the following lemma:

\begin{lemma}
    \label{lem:access}
    Assume \ref{hyp:majoration N}. Let $\ens{A} \in \Xsigma$  be an accessible set for $Q$ such that $\inf_{x \in \ens{A}} \rho_{\tilde R}(x)>0$. Then $\ens{A}$ is accessible for $S$ and $\ens{A} \times \{0\}$ is accessible for $P$.
\end{lemma}

\begin{proof}
    See \ref{appendix:proof:lem:2and3}.
\end{proof}

Let us now assume a drift condition on $Q$ in order to deduce an upper bound for $\sup_{x \in C_{\eta}} \PE_x^S\left[\beta^{\sigma_{\ens{C}_\eta}}
        \right]$.
\newline
\begin{hypn}{H$_{\mbox{\tiny dft}}$}
    \item \label{hyp:drift}
    There exist constants
    $(\eta_0,\lambda)\in (0,1)^2$
    and a measurable function $V:\Xset \to   [1,\infty)$ such that
    \begin{enumerate}%[label=(\alph*)]
        \item for any $x \notin C_{\eta_0}$, we have $QV(x) \leq \lambda V(x)$,
        \item $ b_\infty = \sup_{x \in C_{\eta_0}} \frac{QV}{V}(x) < \infty$,
        \item $ \sup _{x \in C_{\eta_0}} V(x)< \infty$.
    \end{enumerate}
\end{hypn}
\begin{remark}
    Under \ref{hyp:drift}, we have $QV(x) \leq (b_\infty \vee \lambda) V(x)$ for any $x\in \Xset$. A straightforward induction yields for any $n \in\nsetzero$, $1 \leq Q^n V(x)\leq (b_\infty \vee \lambda)^n V(x)$. This implies that $b_\infty \vee \lambda\geq V(x)^{-1/n}$. Since $n$ is arbitrary, $b_\infty \vee \lambda \geq 1$. Finally, $b_\infty \geq 1 > \lambda $ and for all $x\in\Xset$, 
    \begin{equation}
        \label{eq:boundQV}
        QV(x) \leq b_\infty V(x). 
    \end{equation} 
\end{remark}
\begin{lemma} \label{lem:ergo}
    Assume \ref{hyp:majoration N} and \ref{hyp:drift}.
    Then there exist constants $(\eta,\beta_r,\beta_\star) \in (0,\eta_0)
        \times (1,\infty) \times  (0,\infty)$ such that for all $(x,n) \in \Xset\times\nsetzero$ :
    \begin{equation}
        \eset_{(x,n)}^P \left[ \beta_r^{\sigma_{\ens{C}_\eta \times \{ 0\}}} \right]  \leq \beta_\star \beta_r^n V(x)
        \eqsp.
    \end{equation}
\end{lemma}

\begin{proof}
    See \ref{appendix:proof:lem:geom:ergo}.
\end{proof}

Before stating the main theorem, let us introduce a smallness assumption: 
\begin{hypn}{H$_{\mbox{\tiny sml}}$}
    \item \label{hyp:small} For any $\eta \in (0,1)$ there exist a probability measure $\nu\in\measureset_1(\Xset)$ and a constant $\epsilon>0$ such that, $\ens{C}_\eta$ is a $(1,\epsilon\nu)$-small set and    
    $\nu (\ens{C}_\eta\cap\{\tilde R(\,\cdot\,,1)>0\}) >0$.
\end{hypn}    
\begin{theorem}
    \label{thm:ergo}
    Assume \ref{hyp:Qinv}, \ref{hyp:unbiased}, \ref{hyp:majoration N}, \ref{hyp:small} and \ref{hyp:drift}. Then $P$ has a unique invariant probability measure $\bar\pi$ and there exist constants $\delta,\beta_r >1$, $\zeta < \infty$, such that for all $\xi\in\measureset_{1}(\Xset\times\nsetzero)$,

    \begin{equation} \label{eqn:th4}
        \sum_{k=1}^\infty \delta^k d_{TV}( \xi P^k, \bar \pi) \leq \zeta \int_{\Xset\times\nsetzero}\beta_r^n V(x)\, \xi(\rmd x\rmd n).
    \end{equation}

\end{theorem}

\begin{proof}
    According to Theorem 11.4.2 of \cite{doucMarkovChains2018}, there exists $\zeta_0<\infty$ such that:
    \begin{equation} \label{eq:thm4:1}
        \sum_{k=1}^\infty \delta^k d_{TV}( \xi P^k, \bar \pi) \leq \zeta_0 \EE_\xi^P[\beta_r^{\sigma_{\ens{C}_\eta \times \{0\}}}],
    \end{equation}
    provided that:
    \begin{enumerate}[label=(\roman*)]
        \item \label{item:th4:1}  $\ens{C}_\eta  \times \{0\}$ is an accessible $(1,\varepsilon \tilde \nu)$-small set for $P$ satisfying 
        \[\tilde \nu (\ens{C}_\eta \times \{0\})>0,\]
        \item \label{item:th4:2} $\sup_{x \in \ens{C}_\eta} \EE_{(x,0)}^P[\beta_r^{\sigma_{\ens{C}_\eta \times \{0\}}}] < \infty$ for some $\beta_r >1$.
    \end{enumerate}
    Let us start by proving \ref{item:th4:1}. 
    By Lemma \ref{lem:smallset}, there exists $\tilde\nu\in\measureset_{1}(\Xset)$ such that $\ens{C}_\eta  \times \{0\}$ is a $(1,\varepsilon \tilde \nu)$-small set for $P$ satisfying $\tilde \nu (\ens{C}_\eta \times \{0\})>0$. Moreover it is accessible for $P$, by Lemma \ref{lem:access} since $\ens{C}_\eta$ is accessible for $Q$ and $\inf_{x\in \ens{C}_\eta}\rho_{\tilde R}(x)\geq \eta >0$. Hence \ref{item:th4:1}.

    It remains to show \ref{item:th4:2}.
    We can apply Lemma \ref{lem:ergo} to get, for $\beta >1$ :
    \begin{equation*} 
        \sup_{x \in \ens{C}_\eta} \EE_{(x,0)}^P[\beta_r^{\sigma_{\ens{C}_\eta \times \{0\}}}] \leq \beta_\star \sup_{x \in \ens{C}_\eta} V(x) < \infty,
    \end{equation*}
    which shows \eqref{eq:thm4:1}. 
Then \eqref{eqn:th4} is obtained from \eqref{eq:thm4:1} by noting that 
\[\EE_\xi^P[\beta_r^{\sigma_{\ens{C}_\eta \times \{0\}}}] = \int_{\Xset\times\nsetzero}\EE_{(x,n)}^P[\beta_r^{\sigma_{\ens{C}_\eta \times \{0\}}}]\,\xi(\rmd x\rmd n)\]
and applying \Cref{lem:ergo}.

\end{proof}

\section{Pseudo-marginal IMC}

We develop in this section two different frameworks for pseudo-marginal Importance Markov chain \cite{andrieuPseudomarginalApproachEfficient2009}. The first, simplest one is valid if we want to replace $\pi(x)$ by an unbiased estimate $\hat \pi(x)$. This can be written as a specific kernel $\tilde R$ using the same space as classic IMC. 

The second framework tackles the issue of having the intrumental chain $(\tilde X_i)$ being itself a pseudo-marginal chain, i.e. in this case, both $\pi$ and $\tilde \pi$ are computed through two unbiased estimates. 

\subsection{Pseudo-marginal within IMC}

\subsubsection{Adaptation of the kernel $\tilde R$ in the pseudo-marginal setting}

The first  pseudo-marginal approach of the Importance Markov chain can be directly implemented and fits within the framework we develop in this article. Indeed, knowledge of the density of $\pi$ is never assumed, only the unbiasedness assumption \ref{hyp:unbiased} is needed (and the geometric control of hypothesis \ref{hyp:majoration N} for geometric ergodicity).

Assume that for $x\in \Xset$, the density $\pi(x)$ (with respect to some measure $\mu$) is not directly computable but a nonnegative estimate $\hat \pi (x)$ is available, drawn from a kernel $T_{\pi}(x,\cdot)$ such that $\int_{\rset^+}T_{\pi}(x, \rmd w)w = \pi(x)$ Then, one can replace $\varrho_\kappa (x)$  in \eqref{eq:optimal:tildeR} by $\hat \varrho_\kappa(x)  = \kappa \frac{\hat \pi(x)}{\tilde \pi(x)}$ to get a plug-in kernel $\Rpmtilde$ that satisfies \ref{hyp:unbiased}. 
 
This can be formalized as follows. First, define an extended kernel $\tilde R_\psi$ on $\Xset \times \rset^+ \times \mathcal{P} (\nsetzero)$ by
\begin{align*}
    \tilde R_\psi(x,w,\mathrm{d}n ) = (1-\fracpart{\kappa  w / \tilde \pi (x)}) \delta_{\floor{\kappa  w / \tilde \pi (x)}} (\mathrm{d} n)+ 
    \fracpart{\kappa  w / \tilde \pi (x)} \delta_{\floor{\kappa  w / \tilde \pi (x)}+1}(\mathrm{d}n ).
 \end{align*}
Therefore, $\tilde R_\psi (x,\hat  \pi(x),\cdot)$ corresponds to the plug-in random kernel of $\tilde R_{\scriptscriptstyle\mathrm{opt}}$ of \eqref{eq:optimal:tildeR} using the estimate $\hat \pi(x)$. By construction, $ \int_{\nsetzero} n \tilde R_\psi(x,w,\rmd n) = \kappa  w / \tilde \pi (x)$.
We can now define the integrated kernel $\Rpmtilde$ by 

\begin{equation*}
    \Rpmtilde (x,\mathrm{d} n) = \int_{\rset^+} \tilde R_\psi(x,w,\mathrm{d} n) T_{\pi}(x,\mathrm{d}w ), 
\end{equation*}

and $\Rpm$ using \eqref{eq:R}.
\begin{lemma}
    $\Rpmtilde$ satisfies \ref{hyp:unbiased}.
\end{lemma}
\begin{proof}
    Let $x \in \Xset$,
    \begin{align*}
        \int_\nsetzero n \Rpmtilde(x,\mathrm{d}n )= \int_\nsetzero \int_{\rset^+} n \tilde R_\psi (x,w,\mathrm{d}x ) T_\pi(x,\mathrm{d}x ) &= \int_{\rset^+} \left( \int_\nsetzero n \tilde R_\psi(x,w,\mathrm{d}n )\right) T_\pi(x,\mathrm{d} w) \\
        &= \int_{\rset^+} \frac{\kappa w}{\tilde \pi(x)} T_\pi(x,\mathrm{d}w )= \varrho_\kappa(x).
    \end{align*}
\end{proof}

As $\Rpmtilde$ satisfies \ref{hyp:unbiased}, the whole methodology developed in the paper applies to the pseudo marginal case. In particular, the geometric ergodicity result (see next section) still holds if the estimator $\hat  \pi$ is bounded and under \ref{hyp:condition:rejet}, as \ref{hyp:majoration N} will be satisfied.

\subsubsection{Variance of $\Rpmtilde$}

The variance of $\Rpmtilde$ can be expressed as a function of the variance of $\tilde R$. We have the following setup: $W\sim T_\pi(X,\cdot)$ and $N \sim \tilde R_{\psi} (X,W,\cdot)$. Then
\begin{align*}
    \Var(N \vert X) &= \Var \left( \EE[N \vert X,W] \middle \vert X\right) + \EE \left[\Var(N\vert X,W) \middle \vert X \right]\\
    &= \Var\left( \frac{\kappa W}{\tilde \pi(X)} \middle \vert X\right) + \EE [ \Var(N\vert X,W) \vert X ] \\
    &= \frac{\kappa^2 }{\tilde \pi^2(X)}\Var(W \vert X) + \EE \lrb{ \fracpart{\frac{\kappa W}{\tilde \pi(X)}}\lr{1-\fracpart{ \frac{\kappa W}{\tilde \pi(X)}}}\middle\vert X}.    
\end{align*}

The variance can be decomposed into two terms: while the second one is similar to the variance obtained from $\tilde R$ and can also be upper-bounded by $1/4$, the first one is the direct contribution of the variance of the kernel $T$ that generates the estimate. In particular, if $T_\pi(X,\cdot) = \delta_{\pi(X)}$, we recover the same expression as in the previous case.

\subsection{Fully pseudo-marginal IMC}

In this section, we will write that $\pi(\rmd x) = \pi(x) \mu(\rmd x)$ and $\tilde \pi (\rmd x) = \tilde \pi(x) \mu(\rmd x)$ for a common measure $\mu$.

We suppose here that  $(\tilde X_k)$ is itself a pseudo-marginal chain where the estimates of $\tilde \pi$ are drawn from a kernel $T_{\tilde \pi}(x,\cdot)$ such that for all $x\in \Xset$, $\int_{\rset^+} T_{\tilde \pi}(x,\rmd u)u=\tilde \pi(x)$. This construct a two-component Markov chain  $(\tilde X_k,\tilde U_k)$ on  $\Xset \times \rset^+$ that targets $\mu(\rmd x)\, T_{\tilde \pi}(x,\rmd u) u$ \cite{andrieuConvergencePropertiesPseudomarginal2015}.

In order to extend the Importance Markov chain to this case, we need once again to increase the dimension of the chain. The space $\Xset \times \rset^+$ is replaced by $\Xset \times \rset^+ \times \rset^+$. In that case, the second (resp. third) marginal corresponds to the nonnegative estimates of respectively $\tilde \pi $ (resp. $\pi$). The third component $\tilde V_k$ is drawn from a kernel $T_{\pi}(\tilde X_k,\cdot)$ such that for all $x\in \Xset$, $\int_{\rset^+} T_{\pi}(x,\rmd v)v=\pi(x)$. 

This constructs a Markov chain $(\tilde X_k,\tilde U_k, \tilde V_k)$ that targets the probability measure $\tilde \Pi$ defined by 
$$\tilde \Pi(\rmd x\, \rmd u\, \rmd v)= \mu(\rmd x)\, T_{\tilde \pi}(x,\rmd u) u \, T_{\pi}(x,\rmd v).$$

The distribution $\tilde \Pi$ is the \textit{instrumental}  distribution for the extended space. Its first marginal is $\tilde \pi$ as $\tilde \Pi(A\times \rset^+ \times \rset^+)=\tilde \pi(A)$ for any $A\in \Xsigma$. We can define our \textit{target}  distribution $\Pi$ on the same extended space by
$$\Pi(\rmd x\, \rmd u\, \rmd v)= \mu(\rmd x)\, \tilde T_{\tilde \pi}(x,\rmd u)\, T_{\pi}(x,\rmd v)v.$$

We can perform an Importance Markov chain using the instrumental chain  $(\tilde X_k,\tilde U_k, \tilde V_k)$ targeting the instrumental density $\tilde \Pi$ and the target distribution $\Pi$. In this setting, $\varrho_\kappa$ becomes:

$$
\varrho_\kappa(x,u,v)=\kappa \frac{\rmd \Pi}{\rmd \tilde \Pi}(x,u,v)=\kappa \frac{v}{u}.
$$ 

It remains to draw some random integer $\tilde N_k$ with conditional expectation:  

$$\varrho_\kappa(\tilde X_k,\tilde U_k,\tilde V_k)=\kappa \frac{\tilde V_k}{\tilde U_k},$$ 

using for instance $\Ropttilde$.

\section{Numerical experiments}

\subsection{Toy example: mixture of Gaussians}

\subsubsection{Setting}

For starters, we apply Algorithm \ref{algo:imc} on an multidimensional Gaussian mixture. 
The target distribution $\pi$ writes 
\begin{equation*}
    \pi(x) \propto \sum_{i=1}^n \phi_d(x;\mu_i,I_d),
\end{equation*}
where $\phi_d(x; \mu,\Sigma)$ is the density of a Gaussian distribution in dimension $d$ with mean $\mu$ and covariance matrix  $\Sigma$. The means of these distributions are random, i.i.d., $\mu_i \sim \Ncal (0,10^2 I_d)$.
The instrumental distribution $\tilde \pi$ is chosen to be $\tilde \pi(x) \propto \pi(x)^\beta$ for a fixed $\beta \in (0,1)$. The kernel $Q$ targeting $\tilde \pi$ is a No-U-turn Sampler (NUTS) \cite{hoffmanNoUTurnSamplerAdaptively2014}. 

The parameter $\beta$ flattens the instrumental distribution, thus easy out the way for the instrumental chain to move from one mode to another, when compared with the original targetted $\pi$. Conversely, extremely small values of $\beta$ lead to a very flat instrumental distribution from which it is hard to reconstruct the original target $\pi$. Therefore, we test different values of $\beta$ towards finding an optimal tradeoff. The kernel $\tilde R$ used to draw the number of replicas is set to be the same as in  \eqref{eq:optimal:tildeR}, i.e. a shifted Bernoulli.

Note that with such a choice for $\tilde \pi$, the ratio $\frac{\pi}{ \tilde \pi}$ becomes proportional to $\pi^{1 - \beta}$ so that the computation of $\varrho_\kappa$ is just a pointwise evaluation of $\pi$ plus basic float operations. As the choice of $R$ also requires basic float operations (floor, fractional part and multiplication) and a Bernoulli draw, the complexity of Algorithm \ref{algo:imc} is  close to the complexity of directly running $Q$ targeting $\pi$.

To assess the influence of $\beta$ we estimate the mean squared error in approximating the expectation of $\pi$ by running $200$ chains for each $\beta  \in  \{0.004,0.01, 0.04,0.1,1\}$. The results are presented in \Cref{table:MSE:beta}. As expected, for an untempered instrumental distribution, i.e. $\beta=1$, the MSE is high due to a low exploration of the space, and is minimal for $\beta = 0.04$.

\begin{table}[h]
    \begin{center}\begin{tabular}{l|rrrrr}
        {$\beta$} &  0.004 &  0.010 &  0.040 &  0.100 &  1.000 \\
        \hline
        MSE & 16.150 &  6.123 &  \textbf{0.544} & 17.863 & 33.982\\
        \end{tabular}\end{center}
        \caption{\label{table:MSE:beta}MSE for different values of $\beta$ for the Gaussian mixture}        
\end{table}

\subsubsection{Choice of $\kappa$ and analysis of the Effective Sample Size} \label{subsec:choice:kappa}

The parameter $\kappa$ that appears in \eqref{eq:def rho kappa} is somewhat more arbitrary as the other parameters as its value is directly impacted by the normalizing constants of $\pi$ and $\tilde \pi$. Indeed, assume $\pi_U = \pi Z$ (resp. $\tilde \pi_U$) 
is an unnormalized density and write $Z$ (resp. $\tilde Z$ ) the unknown normalization constant. We can then compute the ratio $\varrho_{\kappa,U} := \kappa \frac{\pi_U}{\tilde \pi_U} = \kappa \frac{Z}{\tilde Z} \frac{\pi}{\tilde \pi} = \varrho _{\kappa \frac{Z} {\tilde Z}}$. 
Thus, ignoring the normalizing constants leads to a multiplier term $\kappa$ compared to the case where both densities are normalized.

However it proves possible to overcome this issue by noticing that $\EE [ \sum_i^n \tilde N_i ] = \kappa n$, i.e. that, for an instrumental chain of length $n$, the (expected) length final chain $(X_i)$ is proportional to $\kappa$. So, one way to deal with this issue is to tune $\kappa$ such that the length of the final chain is approximately $\alpha n$, where $\alpha$ is fixed. This is easily solved by taking $\kappa  = \frac{\alpha n}{\sum_{i=1}^n \varrho_U (\tilde X_i)}$, for $\varrho_U = \frac{\pi_U}{\tilde \pi_U}$ .

\begin{remark}
    \label{rmk:diagnosis:kappa}
    The diagnosis of the choice of $\kappa$ can be easily done for a fixed instrumental chain via the computation of the number of replicas for different values of $\kappa$ in parallel is cheap once the vector of the values taken by $\varrho$ is stored. 
\end{remark}

The problem of tuning $\alpha$ remains. We define a metric that may help to this effect, namely, the effective sample size (denoted $\ess_\kappa$, as it depends on $\kappa$). The ESS for an importance Markov chain is similar to the ESS defined for importance sampling, at the difference that here the weights are discretized. Formally,

\begin{equation*}
    \ess_\kappa := \frac{(\sum_{i=1}^n \tilde N_i)^2}{\sum_{i=1}^n \tilde N_i^2}
\end{equation*}
and we also define the usual importance sampling ESS:

\begin{equation*}
    \ess_{\mathrm{IS}} := \frac{(\sum_{i=1}^n \varrho(\tilde X_i))^2}{\sum_{i=1}^n \varrho(\tilde X_i)^2}.
\end{equation*}

\begin{remark}
    Note that the new definition of $\ess_\kappa$ only considers the replication step of the IMC algorithm, and does not take into account the convergence of the instrumental chain to the instrumental target. See \cite{raicescruzIterativeImportanceSampling2022} for a recent work on the effective sample size for IS with dependent proposals. The authors add a term to the $\ess$ to take into account the correlation of the sample.

\end{remark}

Writing $w_i = \frac{\varrho(\tilde X_i)}{\sum_{i=1}^n \varrho(\tilde X_i)}$ for the self-normalized $\varrho$ , we have that $\EE[\tilde N_i \vert \tilde X_i] = \kappa w_i$, so conditionally on $\tilde X_{1:n}$, 
\begin{equation}
    \label{eq:lim:ESS}
    \ess_\kappa \convergesto{\kappa}{\infty} \ess_{\mathrm{IS} },\,\,\,\, \Prob-\text{a.s.}
\end{equation}
As expected by the definition of kernel $R$, as $\kappa$ increases, the stochastic part of the number of replicas vanishes and estimates built with the chain will behave as with importance sampling, in the sense that, for  $M_n := \sum_{i=1}^n \tilde N_i$  : 
\begin{equation*}
    M_n^{-1} \sum_{i=1}^{M_n}h(X_i) = \sum_{i=1}^n   \frac{\tilde N_i}{M_n} h(\tilde X_i)  \convergesto{\kappa}{\infty} \sum_{i=1}^n  w_i h(\tilde X_i).
\end{equation*}

\begin{figure}[H]
    \centering
    \includegraphics[width=0.95\textwidth]{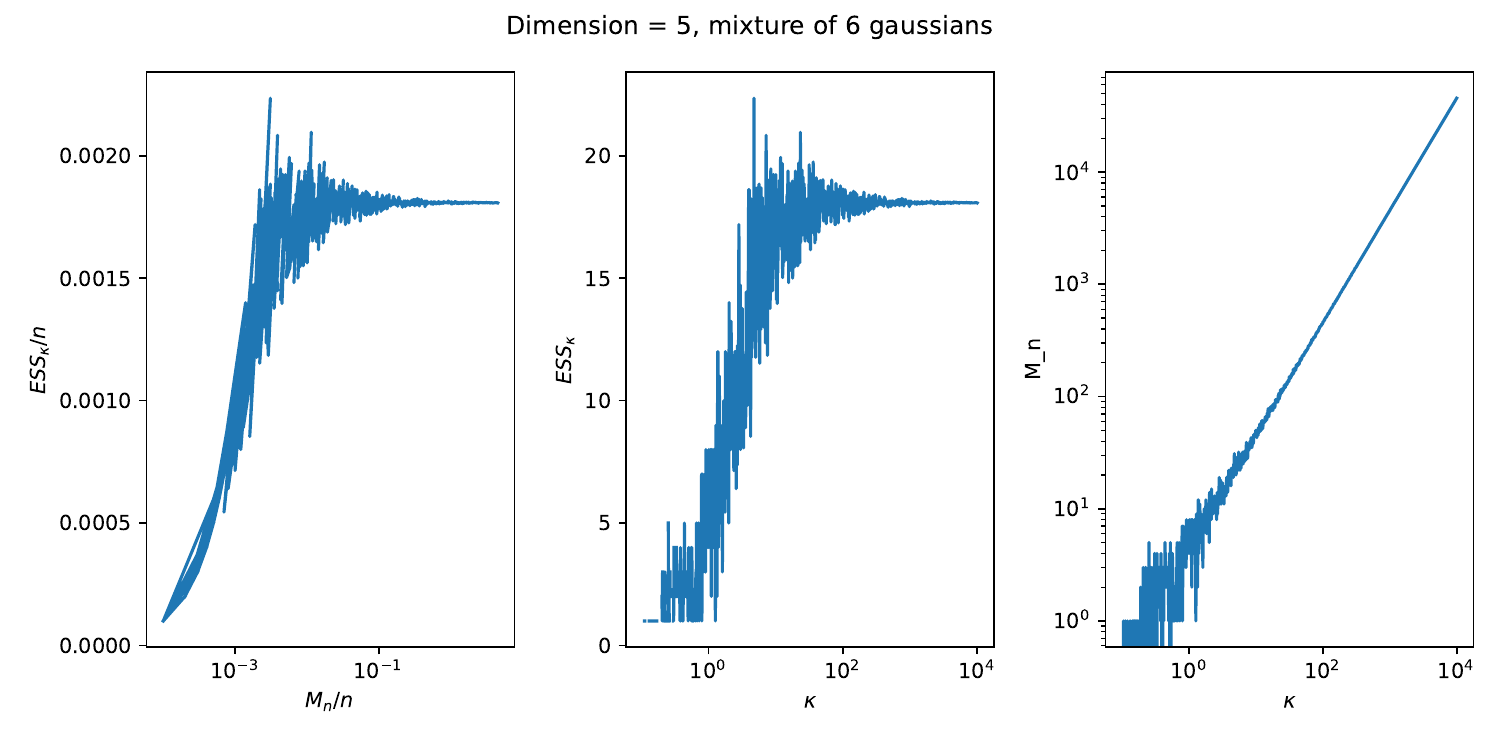}
    \caption{Analysis of the effect of $\kappa$ on the ESS and the length of the chain.
    \label{fig:ESS:kappa}}
\end{figure}

Following \Cref{rmk:diagnosis:kappa}, we plot the effect of $\kappa$ on $\ess_\kappa$ in \Cref{fig:ESS:kappa}. For a fixed instrumental chain of length $n$ and a fixed $\beta = 0.04$, we computed the $\ess_\kappa$  for $10^3$ values of $\kappa$  varying on log scale from $10^{-1}$ to $10^4$. The plot on the right confirms a linear dependence between the length of the final chain and $\kappa$. Both other plots ($\ess_\kappa /n$ as a function of $M_n/n$ and $\ess_\kappa$ as a function of $\kappa$) have the same shape. Overall, the ESS is increasing with $\kappa$ and reaches a stationary regime for $\kappa$ large enough: it converges to $\ess_{\mathrm{IS} }$.  Therefore, taking $\kappa$ too large will not increase the quality of the estimate. In that specific case, this diagnosis leads us to choose $\alpha \simeq 1$.

\subsection{Independent IMC with normalizing flows}

\subsubsection{Settings}

In this section, we compare the Metropolis - Hastings algorithm with independent proposals with the Importance Markov chain algorithm in the specific case where $Q(x,\cdot) = \tilde \pi(\cdot) $ for all $x \in \Xset$.

The target $\pi$ in dimension $d$ is defined, for $x = (x_1,\dots, x_d) \in \rset^d$, by 
\begin{align*}
    \pi(x) \propto &\exp \bigg \{  -\frac{1}{2}\left(\frac{\norm{x}-2}{0.1}\right)^2 + \sum_{i=1}^{d} \log \left( e^{-\frac{1}{2}(\frac{x_i+3}{0.6})^2} + e^{-\frac{1}{2}(\frac{x_i-3}{0.6})^2}\right)\bigg\}.
\end{align*}
This distribution suffers from multimodality  as it has $2$ modes per marginal, for a total of $2^d$ modes. 

The instrumental distribution $\tilde \pi$ is obtained by training a normalizing flow targeting $\pi$. We recall that a normalizing flow is an invertible map $T$ from $\rset^d$ to $\rset^d$. Taking a base distribution $\mu$ on $\rset^d$ (chosen in that case to be the standard Gaussian), the map $T$ should be chosen such that the pushforward measure of $\mu$  through $T$ denoted $T_{\sharp}\mu := \mu(T^{-1}(.))$ is close to $\pi$. This can be done by optimizing $T$ in a family of maps, for instance rational quadratic splines (RQSplines) using neural networks \cite{durkanNeuralSplineFlows2019}. The training is done by minimizing the forward Kullback-Leibler divergence between $T_{\sharp}\mu$ and $\pi$. See \cite{gabrieAdaptiveMonteCarlo2022a} for further details on the training. To get a sample from the flow, one can generate $x\sim \mu$ and derive $T(x)$. The density $\rho$ of $T _\sharp \mu$ is given by, for $x \in \rset^d$:

\begin{equation*}
    \rho(x) = \mu \big(T^{-1}(x)\big) \abs{\det J_{T^{-1}}(x)}.
\end{equation*}
The flows are designed such that $T^{-1}$ and  $J_{T^{-1}}(x)$ are easily computable.

 We used the Python package FlowMC \cite{wongFlowMCNormalizingflowEnhanced2022} with a RQSpline model to train the flow. Every training of a flow yields a different $\tilde \pi = T_\sharp \mu $ as the training is stochastic. For details of implementation, see \ref{sec:details:flows}.
 
 The Self Regenerative Markov Chain Monte Carlo (with no adaptation) \cite{sahuSelfregenerativeMarkovChain2003,gasemyrMarkovChainMonte2002} is close to the Independent IMC, for the special case where the distribution of the number of replicas $\tilde N_i$ is written as 
\[\PP (\tilde N_i =n \vert \tilde X_i = x) = \PP(V S = n),\]
where $V$ is a Bernoulli random variable with parameter $q(x)$ and $S$ is geometric with parameter $\alpha(x)$. In \cite{gasemyrMarkovChainMonte2002}, the author suggests $q(x) = \min(1,\varrho_\kappa(x))$ and $\alpha(x) = \min(1,1/(\varrho_\kappa(x)))$. This method, called \textit{optimal self-regenerative chain} (OSR), is the one we use as a comparison benchmark, with the same tuning of $\kappa$.
 We compare three methods: the independent Metropolis-Hastings (see \cite{robertMonteCarloStatistical2010} for details), the independent importance Markov chain with kernel $\Ropttilde$ of \eqref{eq:optimal:tildeR} and the OSR chain defined above.
 In this case, the IMC algorithm is close to the rejection sampling chain of \cite{tierneyMarkovChainsExploring1994}. 

 The parameter $\kappa$ is tuned such that the length of the final chain is equal to the length of the instrumental chain. The computational cost (and running time) of all three algorithms is similar.

\subsubsection{Results}

\paragraph{Comparison with Metropolis Hastings and OSR}
For each dimension $d \in \{5,10,15,20,25\}$, we trained 10 different flows on the same target. For each flow, 30 i.i.d. samples of size $n = 3 \cdot 10^4$ were generated. 
We display in \Cref{fig:MH.v.IMC.v.OSR} the boxplot of the effective sample size (ESS) of the first marginal for each algorithm, computed using the bulk method.

While the ESS of the OSR is a bit higher than the one obtained with IMH, both are outperformed by the importance Markov chain. Regardless of the dimension, the ESS of the IMC is approximately twice that of MH.
Performances across dimensions are quite similar, even if it is worth noting that the ESS is slightly lower for the dimension 5, probably due to the training of the flow and the fact that the hyperparameters of the flow are not optimized for each dimension.

\begin{figure}[h]
        \centering
        \includegraphics[width=0.8\linewidth]{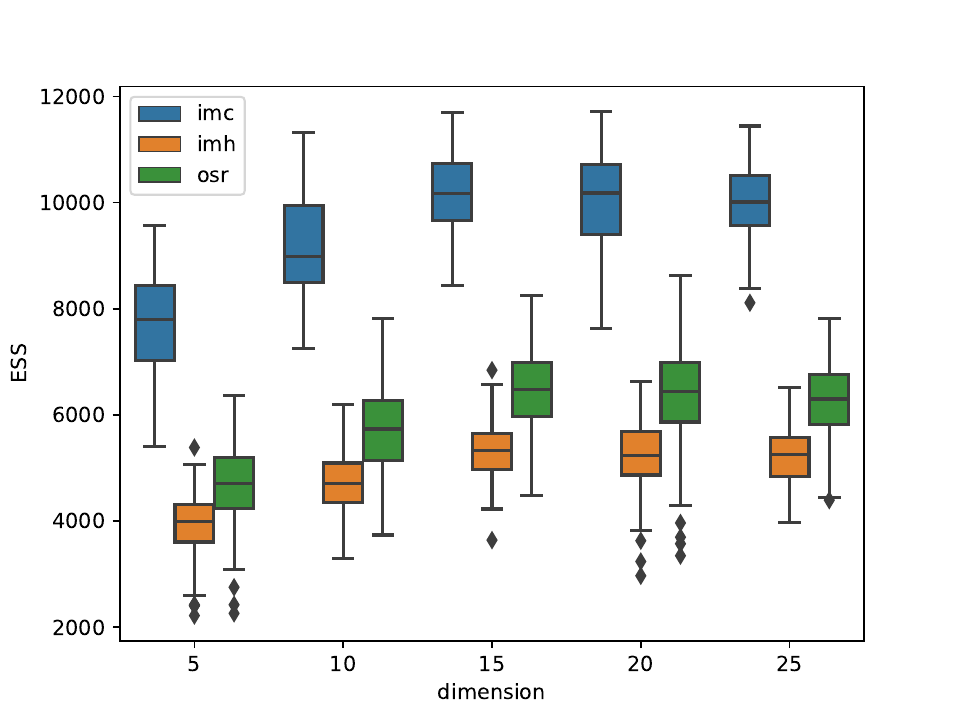}
        \caption{Comparison of the effective sample size between MH, OSR and IMC.}
        \label{fig:MH.v.IMC.v.OSR}
\end{figure}

\paragraph{Comparison with importance sampling}

We compare in \Cref{tab:imc.v.is} the mean squared error (MSE) of the first four odd moments of the first marginal obtained by either the importance Markov chain or with the importance sampling (IS) estimate defined by 
\begin{equation*}
    \hat I_{IS}(h) = \frac{\sum_{i=1}^n \varrho_\kappa(\tilde X_i) h(\tilde X_i)}{\sum_{i=1}^n \varrho_\kappa(\tilde X_i)}
\end{equation*}
against the independent IMC estimate defined by
\begin{equation*}
    \hat I_{IMC} (h)= \frac{1}{k} \sum_{i=1}^{k}h(X_i) = \frac{1}{k}\sum_{i=1}^{n}\tilde N_i h(\tilde X_i),
\end{equation*}
 where we denote $k =\sum_{i=1}^{n} \tilde N_i $ the length of the final chain. The performances of the IMC are very close to the ones of IS, while the gap increases with the dimension. 
The last column shows the mean number points that are replicated once or more by the IMC, which is the length of the chain in the case of importance sampling. This has a strong implication: by storing the importance chain under the representation $(\tilde X_i,\tilde N_i)$, only (around) 16400 points (and their associated number of replicas) are needed to be stored for the dimension 25, instead of $30000$ for importance sampling. If the dimension and the number of points are large, this can be useful to reduce the memory usage. In some contexts, the slight loss in the MSE can be compensated by the gain in memory usage. Moreover, the IMC outputs an actual sample and not a weighted one, sample that approaches the target distribution at mostly a geometric rate.
\begin{table}[H]
    \centering
    \begin{tabular}{|l|l||r|r|r|r|r|}
        \hline
        Dimension & Method & Mean & Third & Fifth  & Seventh  & $\#$ positive copies \\ \hline
        \multirow[c]{2}{*}{5} & IMC & 9.955e-05 & 2.057e-04 & 8.372e-04 & 4.638e-03 & 1.388e+04 \\
        & IS & 9.822e-05 & 2.024e-04 & 8.220e-04 & 4.552e-03 & 3.000e+04 \\ \hline
       \multirow[c]{2}{*}{10} & IMC & 5.176e-05 & 4.146e-05 & 6.799e-05 & 1.653e-04 & 1.512e+04 \\
        & IS & 4.890e-05 & 3.976e-05 & 6.571e-05 & 1.607e-04 & 3.000e+04 \\ \hline
       \multirow[c]{2}{*}{15} & IMC & 3.449e-05 & 1.403e-05 & 1.279e-05 & 1.830e-05 & 1.595e+04 \\
        & IS & 3.229e-05 & 1.341e-05 & 1.254e-05 & 1.804e-05 & 3.000e+04 \\ \hline
       \multirow[c]{2}{*}{20} & IMC & 2.560e-05 & 7.580e-06 & 5.402e-06 & 5.940e-06 & 1.609e+04 \\
        & IS & 2.486e-05 & 7.218e-06 & 5.068e-06 & 5.425e-06 & 3.000e+04 \\ \hline
       \multirow[c]{2}{*}{25} & IMC & 2.356e-05 & 5.585e-06 & 2.962e-06 & 2.483e-06 & 1.646e+04 \\
        & IS & 2.273e-05 & 5.180e-06 & 2.631e-06 & 2.112e-06 & 3.000e+04 \\ \hline
        \end{tabular}
       \caption{Mean square error for the first four odd moments of the first marginal for both the IMC chain and the importance sampling estimate}
       \label{tab:imc.v.is}
\end{table}

\section{Conclusion}

The Importance Markov Chain is a meta-algorithm, in the sense that the produced Markov chain $\seqq{X_k}{k\in\nsetzero}$ is built upon another one, namely $\seqq{\tilde X_k}{k\in\nsetzero}$.
This allows the practitioner to change the target of the MCMC kernel used for the sampling: instead of targeting the distribution of interest directly, our algorithm targets another distribution that may have better properties, and then transforms - with relatively low cost - the obtained sample into an output sample following the distribution of interest. 

A range of applications obviously includes the different instances of Bayesian sensitivity analysis and robustness, such
as a switch of prior distributions in Bayesian robustness \cite{Berger1994, weiss1996}, the inclusion or exclusion of a particular data point for detecting influential observations \cite{jacksonWhiteCarpenter2012}, delayed acceptance MCMC \cite{banterle2019}, safe Bayes modifications of the likelihood \cite{grunwaldVanOmmen2017, Grunwald2023}. A practical illustration is found in the statistical analysis of the standard cosmological model \cite{hilbeDeSouza2017}, where accounting for further satellite observation experiments proves very costly in computing the modified likelihoods.

\section*{Acknowledgments}
Charly Andral is supported by a grant from Région Ile-de-France. Christian Robert is funded by the European Union under the GA 101071601, through the 2023-2029 ERC Synergy grant OCEAN.

\appendix

\section{Postponed proofs}\label{secA1}

The following appendix contains supplementary information that either does not constitute an essential part of the paper, but is helpful in providing a more comprehensive understanding of the research problem, or is too cumbersome to be included in the body of the paper.

\subsection{Uniqueness of the invariant probability measure}
\begin{proof}[(Proof of \Cref{prop:uniqueness})]
    Let $\pi_0$ be an invariant probability measure for $P$ and denote by $\pik{k}$ the measure defined by $\pik{k}(\ens A)=\pi_0(\ens A \times \{k\})$ for any $\ens A \in \Xsigma$. To obtain that $\pi_0=\bar \pi$, we will first express $\pi_0$ using the measure $\pik{0}$ only. Let $f \in \funcset{\Xset}{b+}$. Applying \eqref{eq:def:P} with $h_k(x,n)=f(x)\indiacc{n\neq k}$ and integrating with respect to $\pi_0$ yields for any $k\geq 0$, 
    \begin{equation} \label{eq:recurr:pik}
        \pik{k} f=\pi_0 h_k=\pi_0 Ph_k= \pik{k+1} f + \int_\Xset \pik{0} S(\rmd x) R(x,k) f(x).     
    \end{equation}
    To simplify this equation, we first show that $\pik{0}=\pik{0} S$. Indeed, using again $\pi_0=\pi_0 P$ and \eqref{eq:def:P} with the function $h(x,n)=f(x)$, 
\begin{align*}
    \int_{\Xset\times \nsetzero} \pi_0(\rmd x \rmd n)& f(x)=\PE^P_{\pi_0} [f(X_0)]=\PE^P_{\pi_0} [f(X_1)] = \int_{\Xset \times \nsetzero} \pi_0(\rmd x \rmd n) \indiacc{n \geq 1} f(x) + \int_{\Xset \times \nsetzero} \pi_0(\rmd x \rmd n) \indiacc{n=0} S f(x),  
\end{align*}
which can be equivalently written as $\pik{0} f=\pik{0} S f$. Plugging $\pik{0}=\pik{0} S$ into \eqref{eq:recurr:pik}, we get  
$\pik{k} f= \pik{k+1} f + \int_\Xset \pik{0}(\rmd x) R(x,k) f(x)$ and by straightforward induction, 
$$ 
\pik{k} f=\pik{0}f-\sum_{\ell=0}^{k-1} \int_\Xset \pik{0}(\rmd x) R(x,\ell) f(x)=\int_\Xset \pik{0}(\rmd x) f(x) \lrb{1-\sum_{\ell=0}^{k-1} R(x,\ell)}=\int_\Xset \pik{0}(\rmd x) f(x) \sum_{\ell=k}^{\infty} R(x,\ell). 
$$ 
Hence, for any $h \in \funcset{\Xset \times \nsetzero}{b+}$, 
\begin{align}
    \pi_0 h &=\int_{\Xset \times\nsetzero} \pi_0(\rmd x \rmd n) h(x,n)=\sum_{k=0}^\infty \int_\Xset \pik{k}(\rmd x) h(x,k) \nonumber\\
    &= \sum_{k=0}^\infty  \int_\Xset \pik{0}(\rmd x) h(x,k) \sum_{\ell=k}^{\infty} R(x,\ell)=\sum_{\ell=0}^\infty  \int_\Xset \pik{0}(\rmd x) R(x,\ell)  \sum_{k=0}^{\ell} h(x,k). \label{eq:uniq:zero}
\end{align}
Combining with \eqref{eq:def:pibar}, we can conclude the proof of
\Cref{prop:uniqueness} (ie $\pi_0=\bar \pi$) provided that 
\begin{equation} \label{eq:uniq:two}
    \pik{0}(\rmd x)=\kappa^{-1} \tilde \pi(\rmd x) \rho_{\tilde R}(x). 
\end{equation} 
All that follows consists in proving this identity. 
Denote by $\pi_1$ the measure on $(\Xset \times [0,1],\Xsigma \otimes {\mathcal B}([0,1]))$ defined by: for any function $h \in \funcset{\Xset \times [0,1]}{b+}$,
\begin{equation} \label{eq:defpi1}
    \pi_1 h=\int_{\Xset \times [0,1]} \pik{0}(\rmd x) \rmd u\indi{[0,\rho_{\tilde R}(x)]}(u) \rho_{\tilde R}(x)^{-1} \PE_{(x,u)}^G \lrb{\sum_{k=0}^{\sigma_\ens{D} -1} h(X_k,U_k)},     
\end{equation}
where $\ens D = \sett{(x,u)  \in \Xset \times [0,1]}{u\leq \rho_{\tilde R}(x)}$ and $G$ is defined in \eqref{eq:kernelG}. We first show that $\pi_1=\pi_1 G$.   
\begin{align*}
    \pi_1 G h&=\int_{\Xset \times [0,1]} \pik{0}(\rmd x) \rmd u\ \indi{[0,\rho_{\tilde R}(x)]}(u) \rho_{\tilde R}(x)^{-1} \PE_{(x,u)}^G \lrb{\sum_{k=0}^{\sigma_\ens{D} -1} G h(X_k,U_k)}\\
    &=\sum_{k=0}^\infty \int_{\Xset \times [0,1]} \pik{0}(\rmd x) \rmd u\ \indi{[0,\rho_{\tilde R}(x)]}(u) \rho_{\tilde R}(x)^{-1} \PE_{(x,u)}^G \lrb{ h(X_{k+1},U_{k+1}) \indiacc{k+1 \leq \sigma_\ens{D}}}\\
    &=\int_{\Xset \times [0,1]} \pik{0}(\rmd x) \rmd u\ \indi{[0,\rho_{\tilde R}(x)]}(u) \rho_{\tilde R}(x)^{-1} \PE_{(x,u)}^G \lrb{\sum_{\ell=1}^{\sigma_\ens{D}} h(X_\ell,U_\ell)}.    
\end{align*}
This implies 
\begin{align}
    \pi_1 G h = \pi_1(h)+\int_{\Xset \times [0,1]} \pik{0}(\rmd x) \rmd u\ 
  \indi{[0,\rho_{\tilde R}(x)]}(u) &\rho_{\tilde
    R}(x)^{-1}\PE_{(x,u)}^G \lrb{
    h(X_{\sigma_\ens{D}},U_{\sigma_\ens{D}})
    \indiacc{\sigma_\ens{D} <\infty}} \nonumber \\
    &-\int_{\Xset \times [0,1]} \pik{0}(\rmd x)  \lr{\int_0^{\rho_{\tilde R}(x)}h(x,u) \rmd u} \rho_{\tilde R}(x)^{-1}. 
    \label{eq:uniq:one}
\end{align}
Now, write 
\begin{align*}
    \PE_{(x,u)}^G \lrb{h(X_{\sigma_\ens{D}},U_{\sigma_\ens{D}})\indiacc{\sigma_\ens{D}<\infty} } &=\sum_{\ell =1}^\infty \PE_{(x,u)}^G \lrb{h(X_\ell,U_\ell) \indiacc{U_\ell \leq \rho_{\tilde R}(X_\ell)} \indiacc{\sigma_\ens{D} \geq \ell}} \\
    &=\sum_{\ell =1}^\infty \PE_{(x,u)}^G \lrb{ \lr{\int_0^{\rho_{\tilde R}(X_\ell)} h(X_\ell,u') \rmd u'} \indiacc{\sigma_\ens{D} \geq \ell}}  \\
    &=\sum_{\ell =1}^\infty \PE_{(x,u)}^G \lrb{ \lr{\int_0^{\rho_{\tilde R}(X_\ell)} h(X_\ell,u') \rmd u'} \rho_{\tilde R}(X_\ell)^{-1}\indiacc{U_\ell \leq \rho_{\tilde R}(X_\ell)}\indiacc{\sigma_\ens{D} \geq \ell}} \\
    &=\PE_{(x,u)}^G \lrb{ \lr{\int_0^{\rho_{\tilde R}(X_{\sigma_\ens{D}})} h(X_{\sigma_\ens{D}},u') \rmd u'} \rho_{\tilde R}(X_{\sigma_\ens{D}})^{-1}}\\
    &= \int_\Xset S(x,\rmd x') \lr{\int_0^{\rho_{\tilde R}(x')} h(x',u') \rmd u'} \rho_{\tilde R}(x')^{-1}. 
\end{align*}
Plugging this expression into \eqref{eq:uniq:one} yields: 
$$ 
\pi_1 G h =\pi_1(h) + \int_\Xset \pik{0} S(\rmd x')  \lr{\int_0^{\rho_{\tilde R}(x')}h(x',u') \rmd u'} \rho_{\tilde R}(x')^{-1}-\int_\Xset \pik{0}(\rmd x)  \lr{\int_0^{\rho_{\tilde R}(x)}h(x,u) \rmd u} \rho_{\tilde R}(x)^{-1}=\pi_1 h, 
$$
where we have used that $\pik{0} S=\pik{0}$. Hence $\pi_1$ is an
invariant measure for $G$. Since any invariant measure for $Q$ is
proportional to 
$\tilde \pi$, we deduce from
\eqref{eq:kernelG} that $G$ admits a unique invariant measure (up to a
multiplicative constant) proportional to $\tilde
\pi(\rmd x) \indi{[0,1]}(u) \rmd u$ and hence, $\tilde
\pi(\rmd x) \indi{[0,1]}(u) \rmd u \propto \pi_1$. Now, taking $h(x,n)=\indi{\ens{D}}(x,u) f(x)$ for any arbitrary $f\in \funcset{\Xset}{b+}$, we get, using \eqref{eq:defpi1},
\begin{align*}
    \int_\Xset \tilde \pi(\rmd x) \rho_{\tilde R}(x) f(x) &=\int_{\Xset \times [0,1]} \tilde \pi(\rmd x) \indi{[0,1]}(u) \rmd u\ h(x,u) \\
    &\propto \pi_1 h = \int_{\Xset \times [0,1]} \pik{0}(\rmd x) \rmd u\ \indi{[0,\rho_{\tilde R}(x)]}(u) \rho_{\tilde R}(x)^{-1}  \indi{\ens D}(x,u) f(x)=\pik{0} f.      
\end{align*}
Finally, there exists a constant $\gamma$ such that for any $f\in\funcset{\Xset}{b+}$,  
\begin{equation} \label{eq:uniq:three}
    \pik{0}f= \gamma \int_\Xset \tilde \pi(\rmd x) \rho_{\tilde R}(x) f(x). 
\end{equation}
Applying \eqref{eq:uniq:zero} with $h=\indi{}$ and using sucessively \eqref{eq:R}, \ref{hyp:unbiased} and the identity above, we get 
$$
1=\sum_{\ell=0}^\infty  \int_\Xset \pik{0}(\rmd x) R(x,\ell)  (\ell+1)=\sum_{k=1}^\infty  \int_\Xset \pik{0}(\rmd x) \frac{\tilde R(x,k)k} {\rho_{\tilde R}(x)}=\int_\Xset \pik{0}(\rmd x) \frac{\varrho_\kappa(x)} {\rho_{\tilde R}(x)}=\gamma \int_\Xset \tilde \pi(\rmd x) \varrho_\kappa(x)=\gamma \kappa. 
$$
Combined with \eqref{eq:uniq:three} we finally obtain $\pik{0}f= \kappa^{-1} \int_\Xset \tilde \pi(\rmd x) \rho_{\tilde R}(x) f(x)$ which proves \eqref{eq:uniq:two} and concludes the proof. 
\end{proof}

\subsection{A martingale weak convergence result with random indexes. }

\begin{theorem}
    \label{thm:martingale}
Let $(\Omega,\mcf, \PP)$ be a probability space and let $(\mcf_n)$ be a filtration on $\Omega$ such that $\mcf_n \subset \mcf$ for any $n\in\nsetzero$. 
Let $(M_n)$ be a square-integrable $(\mcf_n)$-martingale such that     
\[
\frac{M_n}{\sqrt{n}} \convlaw{\PP} G\,  \quad \mbox{and} \quad \frac{\EE\lrb{M_n^2}}{n} \rightarrow \sigma^2 
\]
and let $(k_n)$ be a sequence of random integers such that $\frac{k_n}{n}\convproba{\PP} \lambda \in (0,\infty)$. Then  
\[(n\lambda)^{-1/2} M_{k_n} \convlaw{\PP} G.\] 
\end{theorem}

\begin{proof}
Let $\alpha, \lambda^-,\lambda^+$ be positive constants such that $\lambda^- <\lambda <\lambda^+$. 
Define 
\begin{align*}
&B_n=(n\lambda)^{-1/2} M_{k_n},\\
&C_n=(n\lambda)^{-1/2} M_{\floor{n \lambda^-}}, 
\end{align*}
and note that $C_n \weaklim{\PP} (\lambda^-/\lambda)^{1/2} G$.  
Then we have for any $u\in\rset$, 
\begin{align}
    \absm{\EE\lrbm{\rme^{iu B_n}}-\EE\lrbm{\rme^{iu G}}} &\leq\absm{\EE\lrbm{\rme^{iu B_n}}-\EE\lrbm{\rme^{iu C_n}}}+ \absm{\EE\lrbm{\rme^{iu C_n}}-\EE\lrbm{\rme^{iu G}}} \nonumber\\
    &\leq 2\Prob\lr{\absm{B_n - C_n}>\alpha} + \absm{\EE\lrbm{\absm{\rme^{iu B_n}-\rme^{iu C_n}} \indiacc{\absm{B_n -C_n}\leq \alpha}}}  + \absm{\EE\lrbm{\rme^{iu C_n}}-\EE{\rme^{iu G}}} \nonumber \\
    &\leq 2\Prob\lr{\absm{B_n - C_n}>\alpha} +  \sup_{\absm{\beta} \leq \alpha}\absm{\rme^{iu \beta}-1} + \absm{\EE\lrbm{\rme^{iu C_n}}-\EE\lrbm{\rme^{iu G}}}.  \label{eq:one}
\end{align}
Since $(M_i-M_{\floor{n\lambda^-}})_{i \geq \floor{n\lambda^-}}$ is a martingale and $x \to x^2$ is convex,  $\lr{(M_i-M_{\floor{n\lambda^-}})^2}_{i \geq \floor{n\lambda^-}}$ is a non-negative submartingale.  Then, the first term of the right hand side in \eqref{eq:one} may be bounded by applying Doob's maximal inequality to the non-negative submartingale $\lr{(M_i-M_{\floor{n\lambda^-}})^2}_{i \geq \floor{n\lambda^-}}$, 
\begin{align*}
    \Prob(\abs{B_n - C_n}>\alpha) & \leq \Prob\lr{k_n \notin \lrbm{n\lambda^-,n\lambda^+}} + \Prob\lr{\abs{B_n - C_n}>\alpha, k_n \in \lrbm{n\lambda^-,n\lambda^+}} \\
    &\leq \Prob(k_n \notin \lrbm{n\lambda^-,n\lambda^+}) + \Prob\lr{\sup_{i \in \lrbm{n\lambda^-:n\lambda^+}} \absm{M_i - M_{\floor{n \lambda^-}}} >(n \lambda)^{1/2} \alpha }\\
    & \leq \Prob(k_n \notin \lrbm{n\lambda^-,n\lambda^+}) + \frac{\EE\lrb{\lr{M_{\floor{n\lambda^+}} - M_{\floor{n \lambda^-}}}^2}}{n \lambda\alpha^2}\\
    & = \Prob(k_n \notin \lrbm{n\lambda^-,n\lambda^+}) + \frac{\sum_{k=\floor{n\lambda^-}}^{\floor{n\lambda^+}-1}\EE\lrb{\lr{M_{k+1} - M_{k}}^2}}{n \lambda\alpha^2}. 
\end{align*}
Note that since $(M_n)_{n\in\NN}$ is a square-integrable martingale, we have 
\[
    D_n=\frac{\EE[M_n^2]}{n} =\frac{\EE[M_0^2]+\sum_{k=0}^{n-1} \EE[(M_{k+1}-M_k)^2]}{n}, 
\]
and the previous bound writes: 
\[
    \Prob\lr{\abs{B_n - C_n}>\alpha} \leq \Prob(k_n \notin [n\lambda^-,n\lambda^+])+ \frac{\floor{n \lambda^+} D_{\floor{n \lambda^+}}-\floor{n \lambda^-} D_{\floor{n \lambda^-}}}{n \lambda\alpha^2}. 
\]
Finally letting $n$ go to infinity and using successively that $\frac{k_n}{n}\convproba{\PP} \lambda$, $D_n \convergesto{n}{\infty} \sigma^2$ and  $C_n \weaklim{\PP} (\frac{\lambda^-}{\lambda})^{1/2} G$, we obtain
\[
    \limsup_{n \rightarrow \infty} \absm{\EE[\rme^{iu B_n}]-\EE[\rme^{iu G}]} \leq 2\sigma^2 \frac{\lambda^+ - \lambda^-} {\lambda \alpha^2} +  \sup_{\absm{\beta} \leq \alpha}\absm{\rme^{iu \beta}-1} + \absm{\PE[\rme^{iu (\lambda^-/\lambda)^{1/2} G}]-\EE[\rme^{iu G}]}. 
\]
Letting $\lambda^+ \searrow \lambda$ and $\lambda^- \nearrow \lambda$, we get
\[
    \limsup_{n \to \infty}\absm{\PE[\rme^{iu B_n}]-\PE[\rme^{iu G}]} \leq   \sup_{\absm{\beta} \leq \alpha}\absm{\rme^{iu \beta}-1}, 
\] 
and letting $\alpha \to 0$, we finally obtain $\limsup_{n \to \infty}\absm{\PE[\rme^{iu B_n}]-\PE[\rme^{iu G}]}=0$. Therefore $B_n \convlaw{\PP} G$ which concludes the proof. 
\end{proof}

\subsection{Central Limit Theorem}
\label{sec:clt}

\subsubsection{Preliminary results}

Let $\bar Q$ be the Markov kernel on $(\Xset \times\nsetzero) \times (\Xsigma \otimes \mathcal{P}(\nsetzero))$ defined by 
\[
    \bar Q(x,n;\rmd x' \rmd n') = Q(x,\rmd x')\tilde R(x',\rmd n'), 
\]
and $\hat\pi$ be the probability measure on $\Xset\times\NN$ defined by 
    \[
        \hat\pi (\rmd x  \rmd n) = \tilde\pi(\rmd x)\tilde R(x,\rmd n). 
    \]
Let $S_n = \sum_{i=1}^n\tilde N_i$ and define $k_n = \max\lrc{k\in \NN^* : S_k \leq n}$ ensuring that $S_{k_n}\leq n<S_{k_n+1}$.

\begin{lemma}
    \label{lemma:invMeas:Qbar}
    Assume \ref{hyp:Qinv}. Then $\bar Q$ admits $\hat\pi$ as invariant probability measure.
\end{lemma}

\begin{proof}
    Let $\ens A \in \Xsigma \otimes \mathcal{P}(\nsetzero)$. Then
\begin{align*}
    \hat \pi \bar Q(\ens A) &= \int_{(\Xset \times \nsetzero)^2} \tilde \pi (\rmd x ) \tilde R (x,\rmd n) Q(x,\rmd x') \tilde R(x',\rmd n') \indi{\ens A}(x',n')\\
    &= \int_{\Xset\times \nsetzero}\tilde \pi (\rmd x' ) \tilde R (x',\rmd n')\indi{\ens A}(x',n') \\
    &= \hat \pi (\ens A). 
\end{align*} 
\end{proof}
    
\begin{lemma}
    \label{lemma:SLLN:Qbar}
    Assume \ref{hyp:Qinv} and \ref{hyp:SLLN of Q}. 
    Then, for every $\xip \in \measureset_1(\Xset \times \nsetzero)$ and measurable function $g: \Xset \times \nsetzero \rightarrow \rset$ such that $ \hat \pi( \abs{g}) < \infty$,
    \begin{equation}
        \lim_{n \rightarrow \infty} n^{-1} \sum_{k=0}^{n-1} g(\tilde X_k, \tilde N_k) =  \hat \pi(g), \quad \PP_\xip^{\bar Q}-a.s.
    \end{equation}
\end{lemma}

\begin{proof}
    The proof follows that of \Cref{thm: SLLN of P} and also relies on \cite[Proposition 3.5]{doucBoostYourFavorite2022}. 
    Let $\bar h$ be a harmonic function for the kernel $\bar Q$, i.e. for all $(x,n)\in\Xset\times\NN$,
    \[\bar Q \bar h (x,n) = \bar h(x,n), \]
    and let us prove that $\bar h$ is constant. 
    For all $(x,n)\in\Xset\times\NN$,
    \begin{equation}
        \label{lemma:SLLN:Qbar:eq1}
        \bar h(x,n) = \bar Q \bar h(x,n) = \int_{\Xset \times \nsetzero} Q(x,\rmd x')\tilde R(x',\rmd n')\bar h(x',n'). 
    \end{equation}
    The integral on the right hand side of the equation above does not depend on $n$, therefore $\bar h$ is also independant of $n$ and we can write for all $(x,n)\in\Xset\times\NN$,
    \[\bar h (x,n) = \bar h(x,0) =: h_0(x).\]
    Then from \eqref{lemma:SLLN:Qbar:eq1} we have for all $x\in\Xset$,
    \begin{equation}
        \label{lemma:SLLN:Qbar:eq2}
        h_0(x) = \int_{\Xset}Q(x,\rmd x') h_0(x') = Qh_0(x). 
    \end{equation}
    which proves that $h_0$ is harmonic for $Q$. Using \cite[Proposition 3.5]{doucBoostYourFavorite2022}, we get that $h_0$ is constant. Therefore so is $\bar h$ and using \cite[Proposition 3.5]{doucBoostYourFavorite2022} again, we have completed the proof of the lemma. 
\end{proof}

\begin{lemma}
    \label{lm:knn:convprob}
    For all $\zeta \in\measureset_1(\Xset \times \nsetzero)$, $\frac{k_n}{n}\convproba{\PP^{\bar Q}_\zeta}\kappa^{-1}$.
\end{lemma}
\begin{proof}
Let $\beta>\kappa^{-1}$, we will prove that $\Prob^{\bar Q}_\zeta\lr{\frac{k_n}{n}\geq\beta}\convergesto{n}{\infty}0$. From the definition of $k_n$ we have that $\lrcb{k_n\geq\beta n}=\lrcb{\sum^{\floor{\beta n}}_{i=1}\tilde N_i\leq n}$, hence
    \begin{align*}
        \Prob^{\bar Q}_\zeta\lr{\frac{k_n}{n}\geq\beta}
        = \Prob^{\bar Q}_\zeta\lr{\frac{1}{\floor{\beta n}}\sum^{\floor{\beta n}}_{i=1}\tilde N_i\leq \frac{n}{\floor{\beta n}}},
    \end{align*}
    which converges to 0 as $n$ goes to infinity since by applying \Cref{lemma:SLLN:Qbar} with $g(x,\ell)=\ell$ we have
    \begin{equation*}
        \frac{1}{\floor{\beta n}}\sum_{i=1}^{\floor{\beta n}}\tilde N_i \convaszeta \int_{\Xset \times \nsetzero}\tilde\pi(\rmd x)\tilde R(x,\rmd\ell)\ell = \int_{\Xset}\tilde\pi(\rmd x)\varrho_\kappa(x) = \kappa > \beta^{-1}.
    \end{equation*}
    Similarly we prove that for any $\beta<\kappa^{-1}$,
    \[\Prob^{\bar Q}_\zeta\lr{\frac{k_n}{n}<\beta}\convergesto{n}{\infty}0,\]
    by using that $\lrcb{k_n<\beta n}=\lrcb{\sum^{\floor{\beta n}}_{i=1}\tilde N_i> n}$, which completes the proof. 
\end{proof}

\begin{lemma}
    \label{lemma:convProba0:Qbar}
    Assume \ref{hyp:Qinv} and \ref{hyp:SLLN of Q}.
    Let $\zeta\in\measureset_1(\Xset \times \nsetzero)$, $f:\Xset\times\NN\longrightarrow\rset$ be a measurable function, and $(k_n)_n\in\NN^{\NN}$ be a sequence of random variables such that $\frac{k_n}{n}\convproba{\PP^{\bar Q}_\zeta}\kappa^{-1}$ and $\hat\pi f^2<\infty$.  
    Then, 
    \[\frac{f(\tilde X_{k_n}, \tilde N_{k_n})}{\sqrt{n}}\convproba{\PP^{\bar Q}_\zeta}0.\]
\end{lemma}
    
\begin{proof}
Let $\epsilon>0$, we will prove that 
\[\Prob^{\bar Q}_\zeta\lr{g(\tilde X_{k_n}, \tilde N_{k_n})>\epsilon^2n}\convergesto{n}{\infty}0,\]
where $g=f^2$. Let $\alpha,\beta\in\rset^+$ be two nonnegative numbers such that $\alpha<\kappa^{-1}<\beta$ and $\beta-\alpha<\frac{\epsilon^2}{\hat\pi(g)}$. 
\begin{align*}
    \Prob^{\bar Q}_\zeta\lr{g(\tilde X_{k_n}, \tilde N_{k_n})>\epsilon^2n} &\leq \Prob^{\bar Q}_\zeta\lr{\frac{k_n}{n}\notin\lrb{\alpha,\beta}} +
    \Prob^{\bar Q}_\zeta\lr{g(\tilde X_{k_n}, \tilde N_{k_n})>\epsilon^2n, \frac{k_n}{n}\in\lrb{\alpha,\beta}}.  
\end{align*}
From $\frac{k_n}{n}\convproba{\PP^{\bar Q}_\zeta}\kappa^{-1}$ we get $\limsup_n\Prob^{\bar Q}_\zeta\lr{\frac{k_n}{n}\notin\lrb{\alpha,\beta}}=0$ and therefore,
\[\limsup_n\Prob^{\bar Q}_\zeta\lr{g(\tilde X_{k_n}, \tilde N_{k_n})>\epsilon^2n} \leq \limsup_n\Prob^{\bar Q}_\zeta\lr{g(\tilde X_{k_n}, \tilde N_{k_n})>\epsilon^2n, \frac{k_n}{n}\in\lrb{\alpha,\beta}}.\]
Defining $A_n = \frac{1}{n}\sum_{i=1}^{n}g(\tilde X_i, \tilde N_i)$, we have
    \begin{align*}
        \Prob^{\bar Q}_\zeta\lr{g(\tilde X_{k_n}, \tilde N_{k_n})>\epsilon^2n, \frac{k_n}{n}\in\lrb{\alpha,\beta}} &\leq \Prob^{\bar Q}_\zeta\lr{\sum^{\floor{n\beta}}_{i=\floor{n\alpha}}g(\tilde X_i, \tilde N_i)>\epsilon^2n} \\
        &= \Prob^{\bar Q}_\zeta\lr{\frac{\floor{n\beta}}{n} A_{\floor{n\beta}}-\frac{\floor{n\alpha}}{n}A_{\floor{n\alpha}}>\epsilon^2}\convergesto{n}{\infty}0, 
    \end{align*}
since $A_n\convaszeta\hat\pi g$ by \Cref{lemma:SLLN:Qbar} and $\hat\pi(g)(\beta-\alpha)<\epsilon^2$. This concludes the proof.
\end{proof}   

\begin{lemma} \label{lm:martingaleExpr}
    Let $(Y_k)_k$ be a Markov chain on $\Yset$ generated by a kernel $T$ on $\Yset\times\Ysigma$ with invariant measure $\mu$. 
    Assume that there exists a measurable function $j:\Yset\rightarrow\rset$ such that the Poisson equation on $\Yset$ for the kernel $T$ associated to the function $j$ admits a solution $J$, i.e. for all $y\in\Yset$
    \[J(y) - TJ(y) = j(y) - \mu(j).\]
    Then, considering the filtration $\Fcal_n = \sigma(Y_{0:n})$,
    \[\sum_{i=0}^{n-1}j(Y_i) - \mu(j) = M_n - J(Y_n) + J(Y_0) \label{clt:eq:Mn},\]
    where $M_n := \sum_{i=1}^{n}J(Y_i) - TJ(Y_{i-1})$ is a $\Fcal_n$-martingale.
\end{lemma}

\begin{proof} A simple index shift gives that
\[\sum_{i=0}^{n-1}j(Y_i) - \mu(j) = \sum_{i=1}^{n-1}J(Y_i) - TJ(Y_i) 
= \sum_{i=1}^{n}J(Y_i) - TJ(Y_{i-1}) - J(Y_n) + J(Y_1) = M_n - J(Y_n) + J(Y_1).\]
\end{proof}

Let us state as a reminder the following theorem from \cite{doucLimitTheoremsWeighted2008} around which the proof of the following lemma will revolve.

\begin{theorem}[Theorem A.3. of \cite{doucMoulines08}]
    \label{thm:doucmoulines_clt}
    Let $(\Omega,\mcf, \PP)$ be a probability space and let $(\mcf_{n,i})_{i\leq n}$ be a filtration on $\Omega$.
    Assume $\espcond{U_{n,i}^2}{\Fcal_{n,i-1}}<\infty$ for any $n\in\nsetzero$ and any $i=1, \dots,n$, and 
    \begin{align*}
    \label{hyp:doucmoulines_clt_1}
    &\sum_{i=1}^n \left( \espcond{U_{n,i}^2}{\Fcal_{n,i-1}} - \espcond{U_{n,i}}{\Fcal_{n,i-1}}^2 \right) \longrightarrow \sigma^2 \tag{H3} &\text{for some $\sigma>0$}\\
    \label{hyp:doucmoulines_clt_2}
    &\sum_{i=1}^n \espcond{U_{n,i}^2  \indi{\{ \abs{U_{n,i}} \geq \varepsilon \} } }{\Fcal_{n,i-1}}\longrightarrow 0 &\text{for any $\varepsilon>0$} \tag{H4}
    \end{align*}
    Then, for any real $u$, 

    \begin{equation*}
        \espcond{ \exp \left( iu   \sum_{i=1}^n (U_{n,i} - \espcond{U_{n,i}}{\Fcal_{n,i-1}})\right)}{\Fcal_{n,0}} \longrightarrow \exp(-(u^2/2) \sigma^2). 
    \end{equation*}
\end{theorem}

\begin{lemma} \label{lem:martingaleTCL}
    Let $(Y_k)_k$ be a Markov chain on $\Yset$ with kernel $T$ on $\Yset\times\Ysigma$ admitting an invariant probability measure $\mu$. Assume that for every $\nu\in\measureset_1(\Yset)$ and any measurable function $g:\Yset\rightarrow\rset$ such that $\mu(\abs{g})<\infty$,
    \begin{equation} \label{eq:martingale:LLN}
        \lim_{n \rightarrow \infty} n^{-1} \sum_{k=0}^{n-1} g(Y_k) = \mu(g), \quad \PP_\nu^T-a.s.
    \end{equation} 
    Let $J:\Yset\rightarrow\rset$ be a measurable function such that $\mu(J^2)<\infty$. Consider the filtration $\Fcal_n = \sigma(Y_{0:n})$ and the $\Fcal_n$-martingale $M_n = \sum_{i=1}^{n}J(Y_i) - TJ(Y_{i-1}) = \sum_{i=1}^{n}\Delta M_i$. Then,
    \[n^{-1/2}M_n\convlaw{\Prob^M_\nu}\Ncal(0,\sigma_J^2(h)),\]
    with $\sigma_J^2(h)=\EE^T_\mu\lrb{\Delta M_1^2}$.
\end{lemma}
\begin{proof} 
Define $U_{n,i} = \frac{\Delta M_i}{\sqrt{n}}$ for any $n\geq i\geq 1$ and $\Fcal_{n,i} = \Fcal_{i} =\sigma(Y_{0:i})$ for any $i,n\in\NN$. Let us verify the hypotheses of \Cref{thm:doucmoulines_clt} :
\begin{itemize}
    \item For \eqref{hyp:doucmoulines_clt_1}, we have 
    \begin{align*}
        \sum_{i=1}^n \lr{\espcondmkv{\nu}{T}{U_{n,i}^2}{\Fcal_{n,i-1}} - \espcondmkv{\nu}{T}{U_{n,i}}{\Fcal_{n,i-1}}^2} &= \frac{1}{n}\sum_{i=1}^n \espcondmkv{\nu}{T}{\Delta M_i^2}{\Fcal_{i-1}} \\
        &= \frac{1}{n}\sum_{i=1}^n \EE^{T}_{Y_{i-1}}\lrb{\Delta M_1^2} \convergesto{n}{\infty}\EE^{T}_{\mu}\lrb{\Delta M_1^2},
    \end{align*}
    where the limit is obtained from \eqref{eq:martingale:LLN} with the function $g:y\mapsto\EE^{T}_{y}\lrb{\Delta M_1^2}$.
    \item For \eqref{hyp:doucmoulines_clt_2}, let $A>0$ be a positive integer,
    \begin{align*}
        \sum_{i=1}^n \espcondmkv{\nu}{T}{U_{n,i}^2  \indi{\lrcb{\abs{U_{n,i}} \geq \varepsilon}}}{\Fcal_{n,i-1}} &= \frac{1}{n}\sum_{i=1}^n \espcondmkv{\nu}{T}{\Delta M_i^2\indi{\lrcb{\abs{\Delta M_i} \geq \varepsilon \sqrt{n}}}}{\Fcal_{i-1}} \\
        &\leq \frac{1}{n}\sum_{i=1}^n \espcondmkv{\nu}{T}{\Delta M_i^2\indi{\lrcb{\abs{\Delta M_i} \geq A}}}{\Fcal_{i-1}},
    \end{align*}
    for large enough values of $n$. Then, the Markov property gives us that 
    \[\frac{1}{n}\sum_{i=1}^n \espcondmkv{\nu}{T}{\Delta M_i^2\indi{\lrcb{\abs{\Delta M_i} \geq A}}}{\Fcal_{i-1}} = \frac{1}{n}\sum_{i=1}^n \EE^{T}_{Y_{i-1}}\lrbm{\Delta M_1^2\indi{\lrcb{\abs{\Delta M_1} \geq A}}}.\]
    Applying \eqref{eq:martingale:LLN} on the right hand side with $g:y\mapsto\EE^{T}_{y}\lrbm{\Delta M_1^2\indi{\lrcb{\abs{\Delta M_1} \geq A}}}$ gives that 
    \begin{align*}
        \frac{1}{n}\sum_{i=1}^n \EE^{T}_{Y_{i-1}}\lrbm{\Delta M_1^2\indi{\lrcb{\abs{\Delta M_1} \geq A}}}\convasbis \mu g = \EE^{T}_{\mu}\lrbm{\Delta M_1^2\indi{\lrcb{\abs{\Delta M_1} \geq A}}}.
    \end{align*}
    Now let $A\longrightarrow\infty$, $(\EE^{T}_{\mu}\lrbm{\Delta M_1^2\indi{\lrcb{\abs{\Delta M_1} \geq A}}})_A$ converges to 0 by monotone convergence.
    Hence 
    \[\frac{M_n}{\sqrt{n}}\convlaw{\Prob^{T}_{\nu}}\Ncal(0,\sigma_J^2(h)) \text{ with } \sigma_J^2(h)=\EE^{T}_{\mu}\lrb{\Delta M_1^2}.\]  
\end{itemize} 
\end{proof}

\subsubsection{Proof of \Cref{thm:CLT}}
Let $h$ be a bounded measurable function and denote $h_0 := h - \pi h$.
We want to prove that :
\[\frac{1}{\sqrt{n}}\sum_{i=0}^{n-1}h_0(X_i) \convlaw{\Prob^P_{\chi}} \Ncal(0,\sigma^2(h)).\]
Let $U_n = \frac{1}{\sqrt{n}}\sum_{i=0}^{n-1}h_0(X_i)$. Let $\xip=\xi\otimes\mu$ where $\mu$ is any probability measure on $\NN$. The probability measure associated to the trajectories $(X_i)_i$ obtained from the Markov chains $(\tilde X_i, \tilde N_i)_i$ generated by $\bar Q$ starting from $\xip$ is the same as the one associated to the sequence $(X_i)_i$ produced as the first component of the Markov chain $(X_i,N_i)_i$ with kernel $P$ starting from $\chi$ defined by $\chi(f)=\int\xi(\rmd x) S(x,\rmd x')R(x',n')f(x',n')$. We will work with the former probability distribution (i.e. under $\Prob_\xip^{\bar Q}$) as it is best suited for our proof.
Denote $V_n = \frac{1}{\sqrt{n}}\sum_{i=1}^{k_n-1}\tilde N_i h_0(\tilde X_i)$ and let us start by proving that $U_n-V_n\convproba{\Prob_\xip^{\bar Q}} 0$.
By definition of $k_n$, $h_0(\tilde  X_{k_n})$ appears less than $\tilde N_{k_n}$ times in $U_n$ and therefore,
\[\abs{U_n - V_n} = (n - S_{k_n})\frac{\abs{h_0(\tilde X_{k_n})}}{\sqrt{n}}\leq \tilde N_{k_n}\frac{\abs{h_0( \tilde X_{k_n})}}{\sqrt{n}}.\]
From \Cref{lm:knn:convprob}, $\frac{k_n}{n}\convproba{\PP_\xip^{\bar Q}}\kappa^{-1}$ and we can apply \Cref{lemma:convProba0:Qbar} to the function $(x,n)\mapsto n\abs{h_0(x)}$ to obtain $\tilde N_{k_n}\frac{\abs{h_0( \tilde X_{k_n})}}{\sqrt{n}}\convproba{\PP_\xip^{\bar Q}}0$.
Let us write $V_n = \frac{1}{\sqrt{n}}\sum_{i=1}^{k_n-1}f(\tilde X_i, \tilde N_i)$ where $f : (x,n)\mapsto nh_0(x)$.
It now suffices to prove that $V_n\convlaw{\Prob^{\bar Q}_{\xip}}\Ncal(0,\sigma^2(h))$. We will procede in two steps :
\begin{enumerate}[label=(\roman*)]
    \item rewrite $V_n = \frac{1}{\sqrt{n}}M_{k_n} + \delta_n$ using a solution to the Poisson equation associated to $f$, where $(M_n)_{n\in\NN}$ is a martingale and $\delta_n\convproba{\PP_\xip^{\bar Q}}0$ ;\label{clt:enum:i}
    \item prove that $\frac{1}{\sqrt{n}}M_n\convlaw{\Prob^{\bar Q}_{\xip}}\Ncal(0,\sigma_M^2(h))$ 
    and apply \Cref{thm:martingale} to obtain $\frac{1}{\sqrt{n}}M_{k_n}\convlaw{\Prob^{\bar Q}_{\xip}}\Ncal(0,\kappa^{-1}\sigma_F^2(h))$.\label{clt:enum:ii}
\end{enumerate}
Starting with \ref{clt:enum:i}, let $H$ be the solution to the Poisson equation associated to $\varrho h_0$ for the kernel $Q$ on $\Xset$ given by \ref{hyp:CLT}, i.e. for all $x\in\Xset$,
\[H(x) - QH(x) = \varrho(x)h_0(x).\]
Then, $H_\kappa := \kappa H$ is a solution to the Poisson equation associated to $\varrho_\kappa h_0$ for the kernel $Q$ on $\Xset$, and
\begin{equation} \label{eq:Poisson:F}
    F(x,n) := H_\kappa(x) + nh_0(x) - \varrho_\kappa(x)h_0(x)
\end{equation}
 is a solution to the Poisson equation associated to $f$ for the Markov kernel $\bar Q$ such that $\hat\pi F^2<\infty$. Indeed, for $(x,n)\in\Xset\times\NN$ we have
\begin{align*}
    \bar Q F(x,n) &= \int_{\Xset \times \nsetzero} Q(x,\rmd x')\tilde R(x',\rmd n')H(x') + \int_{\Xset \times \nsetzero} Q(x,\rmd x')\tilde R(x',\rmd n')n'h(x') \\
    &\quad \quad - \int_{\Xset \times \nsetzero} Q(x,\rmd x')\tilde R(x',\rmd n')\varrho_\kappa(x')h(x') \\
    &= QH_\kappa(x)
\end{align*}
since $\int_{\nsetzero} \tilde R(x',\rmd n')n' = \varrho_\kappa(x')$, and therefore
\begin{align*}
    F(x,n) - \bar QF(x,n) = H_\kappa(x) + nh_0(x) - \varrho_\kappa(x)h_0(x) - QH_\kappa(x) = nh_0(x) = f(x,n).
\end{align*}
Moreover, $\hat\pi F^2 \leq 4\lr{\kappa^2\hat\pi H^2 + \hat\pi f^2 + \hat\pi(\varrho_\kappa h_0)^2}<\infty$ since $\hat\pi H^2 = \tilde\pi H^2<\infty$ and $\hat\pi f^2<\infty$ from \ref{hyp:CLT}. Then,
\begin{align*}
    \hat\pi f^2&=\int_\Xset\lr{\int_\nsetzero {n^2\tilde R(x,\rmd n)}}h_0(x)^2\tilde\pi(\rmd x) \geq\int_\Xset\lr{\int_\nsetzero n \tilde R(x,\rmd n)}^2h_0(x)^2\tilde\pi(\rmd x) = \int_{\Xset}\varrho_\kappa(x)^2 h_0(x)^2\tilde\pi(\rmd x),
\end{align*}
proving that $\hat\pi\lr{\varrho_\kappa h_0}^2 = \tilde\pi\lr{\varrho_\kappa h_0}^2<\infty$.
Now let $M_n = \sum_{i=2}^{n}F(\tilde X_i, \tilde N_i) - \bar QF(\tilde X_{i-1}, \tilde N_{i-1})$ and consider the filtration $\Fcal_n = \sigma(\tilde X_{1:n}, \tilde N_{1:n})$. From \Cref{lm:martingaleExpr}, $(M_n)_n$ is a $\Fcal_n$-martingale, and
\begin{equation}
    \sum_{i=1}^{n-1}f(\tilde X_i, \tilde N_i) = M_n - F(\tilde X_n, \tilde N_n) + F(\tilde X_1, \tilde N_1) . \label{clt:eq:Mn}
\end{equation}
Therefore, $V_n = \frac{1}{\sqrt{n}}M_{k_n} + \delta_n$
with $\delta_n = \frac{F(\tilde X_{k_n}, \tilde N_{k_n})}{\sqrt{n}} + \frac{F(\tilde X_1, \tilde N_1)}{\sqrt{n}}$. 
Note that $\frac{F(\tilde X_1, \tilde N_1)}{\sqrt{n}}\convproba{\PP_\xip^{\bar Q}}0$ trivially, and $\frac{F(\tilde X_{k_n}, \tilde N_{k_n})}{\sqrt{n}}\convproba{\PP_\xip^{\bar Q}}0$ as a consequence of \Cref{lemma:convProba0:Qbar}. Using \Cref{lem:martingaleTCL} combined to \Cref{lemma:SLLN:Qbar} with $Y_i = (\tilde X_i, \tilde N_i)$,
\[\frac{M_n}{\sqrt{n}}\convlaw{\Prob^{\bar Q}_{\xip}}\Ncal(0,\sigma_F^2(h)) \text{ with } \sigma_F^2(h)=\EE^{\bar Q}_{\hat\pi}\lrb{\Delta M_1^2}.\] 
which proves \ref{clt:enum:ii} and concludes the first part of the proof.
Let us now turn our attention to the expression of the variance claimed by the theorem :
\[\sigma^2(h) = \kappa\tilde\sigma^2(\varrho h_0) + \kappa^{-1}\hat\sigma^2(h_0,\kappa),\]
with $\tilde\sigma^2(\varrho h_0) = 2\tilde\pi\lr{\varrho h_0H} - \tilde\pi\lr{(\varrho h_0)^2}$ and $\hat\sigma^2(h_0,\kappa) = \int_{\Xset}h_0(x)^2 \Var^{\tilde R(x,\cdot)}[N] \tilde\pi(\rmd x)$. From the expression of $F$ in \eqref{eq:Poisson:F} and denoting $\Delta H_1 = H(\tilde X_1)-QH(\tilde X_0)$ we have
    \begin{align}
        \sigma^2(h) &= \kappa^{-1}\EE^{\bar Q}_{\hat\pi}\lrb{\lr{\kappa\Delta H_1+(\tilde N_1-\varrho_\kappa(\tilde X_1))h(\tilde X_1)}^2} \nonumber\\
        &= \kappa\EE^{\bar Q}_{\hat\pi}\lrb{\Delta H_1^2} + \kappa^{-1}\EE^{\bar Q}_{\hat\pi}\lrb{(\tilde N_1-\varrho_\kappa(\tilde X_1))^2h(\tilde X_1)^2} + 2\EE^{\bar Q}_{\hat\pi}\lrb{\Delta H_1(\tilde N_1-\varrho_\kappa(\tilde X_1))h(\tilde X_1)}. \label{eq:sigmah:expression}
    \end{align}
Let us take a look at the first term of the rhs :
\begin{align*}
    \EE^{\bar Q}_{\hat\pi}\lrb{\Delta H_1^2} &= \EE^{\bar Q}_{\hat\pi}\lrb{H(\tilde X_1)^2} + \EE^{\bar Q}_{\hat\pi}\lrb{QH(\tilde X_0)^2} - 2\EE^{\bar Q}_{\hat\pi}\lrb{H(\tilde X_1)QH(\tilde X_0)} \\
    &= \EE^{\bar Q}_{\hat\pi}\lrb{H(\tilde X_0)^2} - \EE^{\bar Q}_{\hat\pi}\lrb{QH(\tilde X_0)^2},
\end{align*}
where we used that $\hat\pi$ is a stationary probability measure.
Noting that $H^2 - QH^2 = (H-QH)(H+QH) = \varrho h _0(2H - \varrho h_0)$, we finally obtain using \ref{hyp:CLT}
\begin{equation}
    \EE^{\bar Q}_{\hat\pi}\lrb{\Delta H_1^2} = 2\tilde\pi\lr{\varrho h_0H} - \tilde\pi\lr{(\varrho h_0)^2}.  \label{eq:CLT:var:2}
\end{equation}
Let us now rewrite the second term of \eqref{eq:sigmah:expression} :
\begin{align*}
    \EE^{\bar Q}_{\hat\pi}\lrb{(\tilde N_1-\varrho_\kappa(\tilde X_1))^2h_0(\tilde X_1)^2} &= \int_\Xset\lr{\int_\nsetzero(n-\varrho_\kappa(x))^2 \tilde R(x,\rmd n)} h_0(x)^2 \tilde\pi(\rmd x) \\
    &= \int_{\Xset}h_0(x)^2 \Var^{\tilde R(x,\cdot)}[N] \tilde\pi(\rmd x), 
\end{align*}
where $\Var^{\tilde R(x,\cdot)}[N] = \int(n-\varrho_\kappa(x))^2 \tilde R(x,\rmd n)$ due to \ref{hyp:unbiased}. For the last term of \ref{eq:sigmah:expression}, the same hypothesis gives that
\begin{align*}
    \EE^{\bar Q}_{\hat\pi}\lrb{\Delta H_1(\tilde N_1-\varrho_\kappa(\tilde X_1))h_0(\tilde X_1)} = \EE^{\bar Q}_{\hat\pi}\lrb{{\EE^{\bar Q}_{\hat\pi}\lrb{\tilde N_1-\varrho_\kappa(\tilde X_1)\vert \tilde X_1,\tilde X_0}}\Delta H_1 h_0(\tilde X_1)} = 0, 
\end{align*}
which concludes the second part of the proof.

\subsection{Geometric ergodicity}

\subsubsection{Proof of \Cref{lem:smallset} and \Cref{lem:access}}
\label{appendix:proof:lem:2and3}

\begin{proof}[Proof of \Cref{lem:smallset}]
    We start with the first point \ref{item:smallset:exist}.
    From the definition of $S$ in \eqref{eq:def:S} and noting that $\ens{C}_\eta$ is a $(1,\varepsilon \nu)$-small set for $Q$, we have for all $x\in \ens{C}_\eta$ and $\ens{A}\in\Xsigma$,
    \begin{align*}
        S(x,\ens{A}) & = \sum_{k =1}^\infty  \eset_x^Q \left[ \rho_{\tilde R} (X_k) \indi{\ens{A}}(X_k)\prod_{i=1}^{k-1} \left( 1 - \rho_{\tilde R}(X_i)\right)  \right]\\
        &\geq Q(x,\rho_{\tilde R}\indi{\ens{A}}) \geq \varepsilon \nu(\rho_{\tilde R}\indi{\ens{A}}).
    \end{align*}
    Applying \eqref{eq:def:P} with $n=0$, we deduce that for all $x\in \ens{C}_\eta$ and all $\ens{B} \in \Xsigma \otimes \mathcal{P}(\nsetzero)$,
    \begin{align*}
        P(x,0;\ens{B})& = \int_{\ens{B}}S(x,\rmd x')R(x',\rmd u') \\
        &\geq \varepsilon \int_{\ens{B}} \nu(\rmd x')\rho_{\tilde R}(x')R(x',\rmd u') \\
        &=: \varepsilon \tilde{\nu}(\ens{B}),
    \end{align*}
    which shows \ref{item:smallset:exist}. 
    Let us now turn to \ref{item:smallset:pos}. 
    Noting that $\ens{C}_\eta^+ := \ens{C}_\eta\cap\{\tilde R(.,1)>0\} = \ens{C}_\eta\cap\{R(.,0)>0\}$ we have,
    \begin{align*}
        \tilde \nu (\ens{C}_\eta \times \{0 \}) &= \int_{\ens{C}_\eta} \nu(\rmd x)\rho_{\tilde R}(x) R(x,0) \\
        &\geq \eta\int_{\ens{C}_\eta} \nu(\rmd x)R(x,0)\\
        &= \eta \int_{\ens{C}_\eta^+} \nu(\rmd x)R(x,0) > 0,
    \end{align*}
    where the last inequality stems from $\nu(\ens{C}_\eta^+)>0$ and $R(x,0)>0$ for all $x\in \ens{C}_\eta^+$. Hence \ref{item:smallset:pos}.
\end{proof}

\begin{proof}[Proof of \Cref{lem:access}]
    Let $\ens{A}$ be accessible for $Q$ such that $\epsilon_\ens{A} \eqdef \inf_{x \in \ens{A}} \rho_{\tilde R}(x)>0$.
    We first show that $\ens{A}$ is  accessible for $S$. Let $x\in \Xset$ and $n \in \nsetzero$ such that $Q^n(x,\ens{A})>0$. Let us 
    consider the representation of $S$ using the kernel $G$ defined in \eqref{eq:kernelG}. Define $\ens{D} \eqdef
        \sett{(x,u) \in  \Xset \times [0,1]}{u \leq \rho_{\tilde R}(x) }$ and $\sigma_{\ens{D}}^{(m)}$ the
    $m$-th return time to the set $\ens{D}$. Then, for any probability measure $\mu$ on $[0,1]$, 
    \begin{align*}
        \int_\Xset Q^n(x,\rmd y) \rho_{\tilde R}(y) \indi{\ens{A}}(y) 
        & = \Prob^G_{\delta_x \otimes \mu} ( (X_n,U_n) \in \ens{D}, X_n \in \ens{A})  \\
            & = \Prob^G_{\delta_x \otimes \mu}  (\exists m \in [1:n], n = \sigma_\ens{D}^{(m)}, X_{\sigma_\ens{D}^{(m)}} \in \ens{A}) \\
            & \leq \Prob^G_{\delta_x \otimes \mu}  (\exists m \in [1:n], X_{\sigma_\ens{D}^{(m)}} \in \ens{A} )               \\
            & = \Prob^S_x (\exists m \in [1:n],X_m \in \ens{A}) \\
            & \leq \sum_{m=1}^n \Prob^S_x (X_m \in \ens{A}).
    \end{align*}
    Hence, $\sum_{m=1}^n S^m(x,\ens{A}) \geq \epsilon_\ens{A} Q^n(x,\ens A)>0$ and there exists an integer $m \leq n$ such that $S^m(x,\ens A)>0$.  Thus, $A$ is an accessible set for $S$.

    We now show that $\ens A \times \{0\}$ is accessible for $P$. Let $(x,k) \in \Xset \times \nsetzero$. From the first part of the proof, there exists $m \in \nsetzero$ such that $S^m(x,\ens A)>0$. 
    Indeed, let $\ens B \eqdef \ens{A}\times \{0\}$, $\ens F:=\Xset\times \{0\}$ and $(Y_n=(X_n,N_n))_{n\in\nsetzero}$ be a Markov chain of kernel $P$. Start by noting that if $x\in\Xset$ and $N\sim R(x,\cdot)$,
    \[\Prob^P_{(x,0)}\lr{ \sigma_\ens{F} < \infty } = \Prob\lr{N < \infty} = 1\]
    since \ref{hyp:majoration N} ensures that $\eset\lrb{N}<\infty$. Then, using the strong Markov inequality, we obtain by induction that 
    \begin{equation*}
        \Prob^P_{(x,0)}\lr{ \sigma^{(m)}_\ens{F} < \infty } = 1.
    \end{equation*}
    Hence,
    \begin{align*}
        \Prob^P_{(x,0)} (Y_{\sigma_\ens{F}^{(m)}} \in \ens B) &=\Prob^P_{(x,0)} (Y_{\sigma_\ens{F}^{(m)}} \in \ens{B}, \sigma_\ens{F}^{(m)} < \infty) \\
        &=\sum_{\ell=1}^{\infty}\Prob^P_{(x,0)} (Y_{\sigma_\ens{F}^{(m)}} \in \ens B, \sigma_\ens{F}^{(m)}=\ell)\\
        &= \sum_{\ell=1}^{\infty}\Prob^P_{(x,0)} (Y_{\ell} \in \ens B, \sigma_\ens{F}^{(m)}=\ell)\\
        &\leq \sum_{\ell=1}^{\infty}\Prob^P_{(x,0)} (Y_\ell \in \ens B)\\
        &= \sum_{\ell=1}^{\infty}P^\ell(x,0; \ens B).
    \end{align*}
    Since by definition of $P$ in \eqref{eq:def:P} we have 
    \begin{equation*}
        \Prob^P_{(x,0)} (Y_{\sigma_\ens{F}^{(m)}} \in \ens B) = \Prob^S_x (X_m \in \ens A) >0,
    \end{equation*}
   at least one of the terms in the sum above is positive and therefore there exists $\ell\in\nsetzero$ such that
    \[P^\ell(x,0\,; \ens A \times\{0\})>0.\]
    Using the definition of $P$ in \eqref{eq:def:P} once more, we have that $P^k(x,k; \{(x,0)\})=1$ and so
    \[P^{k+\ell}(x,k; \ens A \times\{0\})>0,\]
    which concludes the proof.

\end{proof}

\subsubsection{Proof of \Cref{lem:ergo}}
\label{appendix:proof:lem:geom:ergo}
\begin{proof}[(Proof of Lemma \ref{lem:ergo})]
Let $\ens B := {\ens{C}_\eta \times \{ 0\}}$ and $\ens F:=\Xset\times\{0\}$. 
Let $\beta\in(1,\infty)$ be an arbitrary constant and let $D<\infty$ be any positive constant (assuming it exists) such that 
\[\sup_{x\in\Xset}\int_\nsetzero\beta^{n+1} R(x,\rmd n) \leq D.\]
We first show that for all $(x,n)\in\Xset\times\nsetzero$ and $\beta>1$ we have:
\begin{equation}\label{eq:lemma3:1}
    \eset^P_{(x,n)}\lrb{\beta^{\sigma_\ens{B}}} \leq \beta^n\eset_x^S \lrb{D^{\sigma_{\ens{C}_\eta}}}.
\end{equation}
Let us start with the case $n=0$ :
\begin{equation} \label{eq:geom:defEsp}
    \eset_{(x,0)}^P \lrb{ \beta^{\sigma_\ens{B}} } = \sum_{\ell=1}^\infty g^{(\ell)}(x),
\end{equation}
where $g^{(\ell)}(x) := \eset_{(x,0)}^P \lrb{ \beta^{\sigma_\ens{F}^{(\ell)}} \indiacc{\sigma_\ens{F}^{(\ell)}=\sigma_\ens{B}} }$.
\begin{itemize}
    \item For $\ell = 1$,  
    \begin{equation}
        \label{eq:geom:case1}
        g^{(1)}(x) = \eset_{(x,0)}^P \lrb{\beta^{N_1 +1 } \indi{\ens{C}_\eta}(X_{N_1 + 1})} = \int_\Xset\lr{\int_\nsetzero\beta^{n+1} R(x,\rmd n)}\indi{\ens{C}_\eta}(x')S(x,\rmd x') \leq D \times \Prob_x^S\lr{\sigma_{\ens{C}_\eta} = 1},
    \end{equation}
    where the second equality holds from the definition of $P$ in \eqref{eq:def:P} ensuring that $X_{N_1 + 1} = X_1\quad$ $\Prob_{(x,0)}^P$-a.s.

    \item For $\ell > 1$, let us note that $\lrcb{\sigma_\ens{F}^{(\ell)} = \sigma_B} = \lrcb{\sigma_\ens{F}^{(\ell-1)}\circ\theta^{\sigma_\ens{F}} = \sigma_B\circ\theta^{\sigma_\ens{F}}, X_{\sigma_\ens{F}}\notin\ens{C}_\eta}$ as $\Prob_{(x,0)}^P$-a.s. we have $\sigma_\ens{F}^{(\ell)} = \sigma_\ens{F}^{(\ell-1)}\circ\theta^{\sigma_\ens{F}} + \sigma_\ens{F}$, and $\sigma_B = \sigma_B\circ\theta^{\sigma_\ens{F}} + \sigma_\ens{F}$ under the event $\lrcb{\sigma_B>\sigma_\ens{F}}\supset\{\sigma_\ens{F}^{(\ell)} = \sigma_B\}$ for $\ell>1$. 
    Hence,
    \begin{align}
        g^{(\ell)}(x) &= \eset_{(x,0)}^P \lrb{ \beta^{\sigma_\ens{F}^{(\ell)}} \indiacc{\sigma_\ens{F}^{(\ell)}=\sigma_\ens{B}} } \nonumber \\
        &\leq \eset_{(x,0)}^P \lrb{ \beta^{\sigma_\ens{F}} \indi{\bar{\ens{C}}_\eta}(X_{\sigma_\ens{F}})\eset_{(x,0)}^P \lrb{ \beta^{\sigma_\ens{F}^{(\ell-1)}\circ\theta^{\sigma_\ens{F}}} \indiacc{\sigma_\ens{F}^{(\ell-1)}\circ\theta^{\sigma_\ens{F}}=\sigma_\ens{B}\circ\theta^{\sigma_\ens{F}}} \vert \Fcal_{\sigma_\ens{F}} } } \nonumber \\
        &= \eset_{(x,0)}^P \lrb{ \beta^{\sigma_\ens{F}} \indi{\bar{\ens{C}}_\eta}(X_{\sigma_\ens{F}}) \eset_{(X_{\sigma_\ens{F}},0)}^P \lrb{\beta^{\sigma_\ens{F}^{(\ell-1)}}\indiacc{\sigma_\ens{F}^{(\ell-1)}=\sigma_\ens{B}} } } \label{eq:lm5:markov}\\
        &= \eset_{(x,0)}^P \lrb{ \beta^{N_1+1} \indi{\bar{\ens{C}}_\eta}(X_1) \eset_{(X_1,0)}^P \lrb{\beta^{\sigma_\ens{F}^{(\ell-1)}}\indiacc{\sigma_\ens{F}^{(\ell-1)}=\sigma_\ens{B}} } }, \label{eq:lm5:split}
    \end{align}
\end{itemize}
where \eqref{eq:lm5:markov} comes from the strong Markov property applied to the Markov chain $(X_i,N_i)_{i\in\nsetzero}$ with the stopping time $\sigma_\ens{F}$ and \eqref{eq:lm5:split} comes from the definition of the kernel $P$ since $X_1$ is repeated $N_1 + 1$ times while the second component decreases by one at each iteration until reaching zero. The definition of $P$ in \eqref{eq:def:P} gives
\begin{align*}
    \eset_{(x,0)}^P \lrb{ \beta^{N_1+1} \indi{\bar{\ens{C}}_\eta}(X_1) \eset_{(X_1,0)}^P \lrb{\beta^{\sigma_\ens{F}^{(\ell-1)}}\indiacc{\sigma_\ens{F}^{(\ell-1)}=\sigma_\ens{B}} } } &= \int_\Xset\lr{\int_\nsetzero\beta^{n+1} R(x,\rmd n)}\indi{\bar{\ens{C}}_\eta}(x')g^{(\ell-1)}(x')S(x,\rmd x') \\
    &\leq D \int_\Xset \indi{\bar{\ens{C}}_\eta}(x') g^{(\ell-1)}(x')S(x,\rmd x') \\
    &= D\times\eset_{(x,0)}^P \lrb{\indi{\bar{\ens{C}}_\eta}(X_1) g^{(\ell-1)}(X_1) }.
\end{align*}
Hence, 
\begin{equation*}
    g^{(\ell)}(x) = \eset_{(x,0)}^P \lrb{ \beta^{\sigma_\ens{F}^{(\ell)}} \indiacc{\sigma_\ens{F}^{(\ell)}=\sigma_\ens{B}} } \leq D\times\eset_{(x,0)}^P \lrb{ \indi{\bar{\ens{C}}_\eta}(X_{\sigma_\ens{F}}) g^{(\ell-1)}\lr{X_{\sigma_\ens{F}}} },
\end{equation*}
which used in conjunction with \eqref{eq:geom:case1} gives 
\begin{align}
    g^{(\ell)}(x) &\leq D^\ell \times \eset_{(x,0)}^P \lrb{ \indi{\bar{\ens{C}}_\eta}(X_{\sigma_\ens{F}})...\indi{\bar{\ens{C}}_\eta}(X_{\sigma^{(\ell-1)}_\ens{F}}) \Prob_{X_{{\sigma_\ens{F}}^{(\ell-1)}}}^S\lr{\sigma_{\ens{C}_\eta} = 1} } \nonumber \\
    &= D^\ell \times \eset_x^S \lrb{ \indi{\bar{\ens{C}}_\eta}(X_1)...\indi{\bar{\ens{C}}_\eta}(X_{\ell-1}) \eset_{X_{\ell-1}}^S\lrb{\indi{\ens{C}_\eta}(X_1)} } \label{eq:geom:PversS}  \\
    &= D^\ell \times \eset_x^S \lrb{ \indi{\bar{\ens{C}}_\eta}(X_1)...\indi{\bar{\ens{C}}_\eta}(X_{\ell-1}) \eset_x^S\lrb{\indi{\ens{C}_\eta}(X_\ell) \vert \Fcal_{\ell-1}} } \nonumber \\
    &= D^\ell \times \eset_x^S \lrb{ \indi{\bar{\ens{C}}_\eta}(X_1)...\indi{\bar{\ens{C}}_\eta}(X_{\ell-1})\indi{\ens{C}_\eta}(X_\ell) } \nonumber \\
    &= D^\ell \times \Prob_x^S\lr{\sigma_{\ens{C}_\eta} = \ell}, \nonumber
\end{align}
 where \eqref{eq:geom:PversS} comes from the definition of $S$ and $P$ in \eqref{eq:def:S} and \eqref{eq:def:P}. Plugging the above into \eqref{eq:geom:defEsp} leads to the following upper bound,
\begin{align*}
    \eset_{(x,0)}^P \lrb{ \beta^{\sigma_\ens{B}} } \leq \sum_{\ell=1}^\infty D^\ell \Prob_x^S \lr{\sigma_{\ens{C}_\eta}=\ell} = \eset_x^S \lrb{D^{\sigma_{\ens{C}_\eta}}}.
\end{align*}

If $x\in \ens{C}_\eta$, then $\eset^P_{(x,n)}\lrb{\beta^{\sigma_\ens{B}}} = \beta^n$, showing \eqref{eq:lemma3:1} for $x\in \ens{C}_\eta$.
If $x\notin \ens{C}_\eta$, then we have $\sigma_\ens{B} = \sigma_\ens{B}\circ\theta^n+n$, $\Prob_{(x,n)}$-a.s. Thus,
\begin{align*}
    \eset^P_{(x,n)}\lrb{\beta^{\sigma_\ens{B}}} &= \beta^n\eset^P_{(x,n)}\lrb{\beta^{\sigma_\ens{B}\circ\theta^n}} = \beta^n\eset^P_{(x,0)}\lrb{\beta^{\sigma_\ens{B}}} \leq \beta^n\eset_x^S \lrb{D^{\sigma_{\ens{C}_\eta}}},
\end{align*}
which completes the proof of the inequality \eqref{eq:lemma3:1}.

For any $\eta \in (0,\eta_0)$,
\begin{equation*}
        \PP_x^S \lr{ \sigma_{\ens{C}_\eta} >k } = 
        \sum_{\ell_1,\dots,\ell_k =1}^\infty  \PE_x^Q \left[ \prod_{j=1}^{k}\left[ \prod_{i = s_{j-1}+1}^{s_j-1} \left(1-\rho_{\tilde R}(X_i)\right)  \right] \rho_{\tilde R}(X_{s_j}) \indi{\bar{\ens{C}}_\eta} (X_{s_j})\right],
\end{equation*}
where we have set $s_j=\sum_{i=1}^j \ell_i$ for $j \geq 1$ and $s_0=0$. Using that $ \rho_{\tilde R}
\indi{\bar{\ens{C}}_\eta}  \leq \eta \indi{\bar{\ens{C}}_\eta} $ and $1-\rho_{\tilde R} \leq (1-\eta_0)^{\indi{\ens{C}_{\eta_0}}}$, 
\begin{align}
    \PP_x^S \lr{\sigma_{\ens{C}_\eta} >k} &\leq \sum_{\ell_1,\dots,\ell_k =1}^\infty \eta^k\PE_x^Q \left[
        (1-\eta_0)^{\sum_{j=1}^{k} \sum_{i = s_{j-1}+1}^{s_j-1} \indi{\ens{C}_{\eta_0}}(X_i)} 
     \prod_{j=1}^k \indi{\bar{\ens{C}}_\eta} (X_{s_j})\right] \nonumber \\
     &\leq  \sum_{\ell_1,\dots,\ell_k =1}^\infty \eta^k\PE_x^Q \left[
        (1-\eta_0)^{M_{[1:s_k-1]\setminus \{s_1,\ldots,s_{k-1}\}}} \prod_{j=1}^{k-1} \indi{\bar{\ens{C}}_\eta} (X_{s_j})\right], \label{eq:geom:zero}
\end{align}
where $M_I\eqdef \sum_{i \in I} \indi{\ens{C}_{\eta_0}}(X_i)$ for any $I \subset \nsetzero$. Moreover, since $\eta <\eta_0$, we have $\ens{C}_{\eta_0} \subset \ens{C}_\eta$, hence  
\begin{equation}\label{eq:geom:one}
    \indi{\bar{\ens{C}}_\eta} (x)=(1-\eta_0)^{\indi{\ens{C}_{\eta_0}} (x)}\indi{\bar{\ens{C}}_\eta} (x) \leq (1-\eta_0)^{\indi{\ens{C}_{\eta_0}} (x)}.    
\end{equation}
Finally, setting for any arbitrary $ \alpha \in (0,1)$,   
\begin{align*}
    & N_\ell=\sum_{i=1}^{\ell-1} \indi{\ens{C}_{\eta_0}}(X_i),\\
    &R_{\ell}\eqdef \PE_x^Q \lrb{(1-\eta_0)^{N_{\ell}}},\\
    &R_{\ell,1} = \PE_x^Q \lrb{(1-\eta_0)^{N_{\ell}} \indiacc{N_{\ell} > \alpha(\ell-k) }},\\
    &R_{\ell,2}=  \PE_x^Q \lrb{(1-\eta_0)^{N_{\ell}}  V(X_{\ell})\indiacc{N_{\ell} \leq \alpha(\ell-k) }},   
\end{align*}
and plugging \eqref{eq:geom:one} into \eqref{eq:geom:zero} combined with $V\geq 1$, we get
\begin{equation}
    \PP_x^S \lr{\sigma_{\ens{C}_\eta} >k} \leq \sum_{\ell_1,\dots,\ell_k =1}^\infty \eta^k R_{s_k} \leq \sum_{\ell_1,\dots,\ell_k =1}^\infty \eta^k \lrcb{R_{s_k,1} +R_{s_k,2}} \leq \sum_{\ell_1,\dots,\ell_k =1}^\infty \eta^k \lrb{(1-\eta_0)^{\alpha(s_k-k)} +R_{s_k,2}}. \label{eq:geom:three}
\end{equation}
We now give an explicit upper bound for $R_{\ell,2}$ for $\ell\geq k$. Under \ref{hyp:drift}, $QV(x) \leq \lambda V(x)$ for $x\notin \ens{C}_{\eta_0}$ and $QV(x) \leq b_\infty V(x)$ for $x\in \ens{C}_{\eta_0}$. Therefore, for any $x\in \Xset$, 
$$ 
 (1-\eta_0)^{\indi{\ens{C}_{\eta_0}}(x)} QV(x) \leq (b_\infty (1-\eta_0) \lambda^{-1})^{\indi{\ens{C}_{\eta_0}}(x)} \lambda V(x) \leq A^{\indi{\ens{C}_{\eta_0}}(x)} \lambda V(x),
$$ 
where $A:= 1 \vee  b_\infty(1-\eta_0) \lambda^{-1}$. This implies by applying the tower rule conditionally on $X_{1:\ell-1}$, then $X_{1:\ell-2}$ and so on,  
\begin{equation}
    \PE_x^Q\lrb{V(X_\ell) \frac{(1-\eta_0)^{N_\ell} }{A^{N_\ell} \lambda^{\ell-1}} }=\PE_x^Q\lrb{V(X_1)\prod_{i=1}^{\ell-1} \frac{(1-\eta_0)^{\indi{\ens{C}_{\eta_0}}(X_{i})} V(X_{i+1})}{A^{\indi{\ens{C}_{\eta_0}}(X_{i})} \lambda V(X_{i})} }\leq \PE_x^Q\lrb{V(X_1)}=QV(x)\leq b_\infty V(x), \label{eq:geom:four}
\end{equation}
where the last inequality follows from \eqref{eq:boundQV}. Since $A\geq 1$, we have $1\leq A^{\alpha(\ell-k)}/A^{N_{\ell}}$ on $\{N_{\ell} \leq \alpha(\ell-k)\}$ and hence
\begin{equation*}
    R_{\ell,2}\leq   A^{\alpha(\ell-k)} \lambda^{\ell -1}\PE_x^Q\lrb{V(X_\ell) \frac{(1-\eta_0)^{N_\ell} }{A^{N_\ell} \lambda^{\ell-1}} } \leq A^{\alpha(\ell-k)} \lambda^{\ell -1} b_\infty V(x),
\end{equation*}
where the last inequality follows from \eqref{eq:geom:four}. 
Pick $\alpha$ small enough so that $\lambda A^\alpha <1 $. Plugging the inequality above (with $\ell$ replaced by $s_k=\ell_1+\cdots+\ell_k$) into \eqref{eq:geom:three} yields 
\begin{align*}
    \PP_x^S \lr{\sigma_{\ens{C}_\eta} >k}  & \leq \sum_{\ell_1,\dots,\ell_k =1}^\infty \eta^k \left[  (1- \eta_0)^{\alpha(\ell_1 + \dots + \ell_k -k)} + \left[ \lambda A^\alpha \right]^{\ell_1 + \dots + \ell_k}A^{-\alpha k}\lambda^{-1} b_\infty  V(x) \right]                          \\
    & = \eta^k \left[ \left( \frac{(1- \eta_0)^\alpha}{1 - (1 - \eta_0)^\alpha} \right)^k \frac{1}{(1-\eta_0)^{\alpha k}} + \left( \frac{\lambda A^\alpha}{1 - \lambda A^\alpha} \right)^k \frac{1}{A^{\alpha k}} \lambda^{-1} b_\infty  V(x)\right] \\
    & = \eta^k \lrb{\lr{\frac{1}{1 - (1-\eta_0)^\alpha}}^k + \lr{\frac{ \lambda}{1-\lambda A ^\alpha}} ^k \lambda^{-1} b_\infty  V(x)}.
\end{align*}
Now set $\gamma  := \max\lr{\frac{1}{1 - (1-\eta_0)^\alpha},\frac{\lambda}{1-\lambda A ^\alpha}}$ and choose $\eta < \eta_0$ sufficiently small so that $\eta \gamma  <1$.  Then, 
\begin{equation*}
    \PP_x^S  \lr{\sigma_{\ens{C}_\eta} >k} \leq \eta^k \gamma ^k(1+\lambda^{-1} b_\infty  V(x)),
\end{equation*}
and if $D \in (1,\eta^{-1}\gamma ^{-1})$, 
\begin{align*}
    \PE_x^S \left[ \frac{D^{\sigma_{\ens{C}_\eta}}-1}{D-1} \right]&= \PE_x^S \left[ \sum_{k=0}^\infty D^k \indiacc{k<\sigma_{\ens{C}_\eta}} \right] =\sum_{k=0}^\infty D^k  \PP_x^S \left( \sigma_{\ens{C}_\eta} >k  \right)\\
&\leq \sum_{k=0}^\infty D^k \gamma ^k \eta^k (1+\lambda^{-1} b_\infty  V(x)) = \frac{1+\lambda^{-1} b_\infty  V(x)}{1-D \gamma  \eta}.
\end{align*}
From \ref{hyp:majoration N} there exists $\beta_0\in(1,\infty)$ and $D_0<\infty$ such that 
\[\sup_{x\in\Xset}\int_\nsetzero\beta_0^{n+1} R(x,\rmd n)\leq D_0.\]
Let $r\in(0,1)$ and consider $\beta_r=\beta_0^r$. From Hölder's inequality,
\begin{equation}
    \int_\nsetzero\beta_r^{n+1} R(x,\rmd n) \leq \lr{\int_\nsetzero\beta_0^{n+1} R(x,\rmd n)}^r \leq D_0^r.
\end{equation}
Choose $r$ such that $D_r := D_0^r\in (1,\eta^{-1}\gamma ^{-1})$. We can now apply the above inequality in combination with \eqref{eq:lemma3:1} using the couple $(\beta_r, D_r)$ instead of $(\beta,D)$ and obtain for all $(x,n)\in\Xset\times\nsetzero$,

\begin{align*}
    \eset^P_{(x,n)}\lrb{\beta_r^{\sigma_\ens{B}}} \leq \beta_r^n \PE^S_x \lrb{ D_r^{\sigma_{\ens{C}_\eta}} } & \leq \beta_r^n\lrb{1 +   (D_r-1) \frac{1 + \lambda^{-1} b_\infty  V(x)}{1 - D_r \gamma  \eta } }\leq \beta_\star \beta_r^n V(x)
\end{align*}
since $V\geq 1$, and where 
\begin{equation*}
    \beta_\star = 1+(D_r-1)\frac{1+\lambda^{-1}b_\infty}{1 - D_r \gamma  \eta }.
\end{equation*}
The proof is completed. 

\end{proof}

\section{Numerical experiments}
\subsection{Details of the hyperparameters for the normalizing flows}
\label{sec:details:flows}
The normalizing flow is a RQSpline with $10$ layers, $8$ bins, and a $(128,128)$ hidden size. The local sampler is a MALA algorithm with step size $0.1$. Training consists of a total of 30 loops with a unique epoch each time. 10 global steps are implemented with a further 10 local steps between each, and the optimizer (Adam) has a learning rate of $8 \cdot 10^{-4}$ and a momentum of 0.9. The seed is 1250. The code generating the figures is available at \url{https://github.com/charlyandral/importance_markov_chain}.


%% References


\bibliographystyle{elsarticle-num}
\bibliography{spa_biblio.bib}






% \end{thebibliography}

\end{document}